\documentclass{lipics-v2021-arxiv}

\usepackage[appendix=append,bibliography=common]{apxproof}

\usepackage{enumerate}
\usepackage{mathtools}
\usepackage{nicefrac}
\usepackage{tikz}
\usetikzlibrary{automata,arrows,positioning}
\usepackage{bm}
\usepackage[all]{xy}
\usepackage{url}
\usepackage[most]{tcolorbox}
\usepackage{subcaption}

\usepackage[Algorithmus]{algorithm}
\usepackage[end]{algpseudocode}

\allowdisplaybreaks

\usepackage[framemethod=tikz]{mdframed}
\newmdenv[innerlinewidth=0.2pt,
roundcorner=4pt,linecolor=black,backgroundcolor=orange!70,
innerleftmargin=6pt,
innerrightmargin=6pt,innertopmargin=6pt,innerbottommargin=6pt]{mybox}

\usepackage{amssymb,amsmath,stmaryrd}

\usepackage{mathabx}

\newtheoremrep{theorem}{Theorem}[section]
\newtheoremrep{lemma}[theorem]{Lemma}
\newtheoremrep{definition}[theorem]{Definition}
\newtheoremrep{proposition}[theorem]{Proposition}
\newtheoremrep{corollary}[theorem]{Corollary}
\newtheoremrep{example}[theorem]{Example}

\renewcommand{\int}{\ensuremath{\mathbb{Z}}}

\newcommand{\onehalf}{\nicefrac{1}{2}}

\newcommand{\one}{\ensuremath{1}}
\newcommand{\zero}{\ensuremath{0}}

\newcommand{\comp}[1]{\ensuremath{\overline{#1}}}

\newcommand{\inc}[2]{\iota_{#1}^{#2}}

\newcommand{\Pow}[1]{\ensuremath{\mathcal{P}(#1)}}
\newcommand{\Powf}[1]{\ensuremath{\mathcal{P}_{\mathit{fin}}(#1)}}
\newcommand{\Ytop}[2]{[{#1}]^{#2}}

\newcommand{\interval}[2]{\ensuremath{[{#1}, {#2}]}}

\newcommand{\len}[1]{\ensuremath{|{#1}|}}

\newcommand{\norm}[1]{\ensuremath{|\!|{#1}|\!|}}

\newcommand{\Min}{\textsf{Min}}
\newcommand{\Max}{\textsf{Max}}

\newcommand{\mins}{\min\nolimits}
\newcommand{\maxs}{\max\nolimits}

\makeatletter
\newcommand{\superimpose}[2]{%
  {\ooalign{$#1\@firstoftwo#2$\cr\hfil$#1\@secondoftwo#2$\hfil\cr}}}
\makeatother

\newcommand{\minimp}[1]{\ensuremath{imp_{\min}}({#1})}
\newcommand{\minimps}[1]{\ensuremath{imp^s_{\min}}({#1})}
\newcommand{\maximp}[1]{\ensuremath{imp_{\max}}({#1})}
\newcommand{\maximps}[1]{\ensuremath{imp^s_{\max}}({#1})}

\newcommand{\monM}{\mathbb{M}}

\newcommand{\latL}{\mathbb{L}}

\definecolor{dmagenta}{rgb}{0.81,0,0.81}
\definecolor{dcyan}{rgb}{0,0.6,0.6}
\definecolor{dgreen}{rgb}{0.09, 0.45, 0.27}
\definecolor{dred}{rgb}{0.55, 0.0, 0.0}

\newcommand{\blue}{\color{blue}}
\newcommand{\red}{\color{dred}}

\newcommand{\pb}[1]{{\blue #1}}
\newcommand{\mystrutab}{\raisebox{0ex}[0.4cm]{}}
\newcommand{\mystrutbl}{\raisebox{-2ex}[0.4cm]{}}

\newcommand{\mathup}[1]{\text{}}
\newcommand{\state}[1]{\ensuremath{\bm{#1}}}

\title{A Lattice-Theoretical View of Strategy Iteration}
\titlerunning{A Lattice-Theoretical View of Strategy Iteration}

\author{Paolo Baldan}{Universit\`a di Padova, Italy}{baldan@math.unipd.it}{https://orcid.org/0000-0001-9357-5599}{}
\author{Richard Eggert}{University of Duisburg-Essen, Germany}{richard.eggert@uni-due.de}{https://orcid.org/0000-0002-9901-7392}{}
\author{Barbara K\"onig}{University of Duisburg-Essen, Germany}{barbara\_koenig@uni-due.de}{https://orcid.org/0000-0002-4193-2889}{}
\author{Tommaso Padoan}{Universit\`a di Padova, Italy}{padoan@math.unipd.it}{https://orcid.org/0000-0001-7814-1485}{}

\authorrunning{P. Baldan and R. Eggert and B. K\"onig and T. Padoan}

\funding{Partially supported by the DFG project SpeQt and by the Ministero dell’Universt\`a e della Ricerca Scientifica of Italy, under Grant No. 201784YSZ5, PRIN2017 -- ASPRA.}

\Copyright{Paolo Baldan and Richard Eggert and Barbara K\"onig and Tommaso Padoan}



\nolinenumbers %uncomment to disable line numbering

\begin{document}

\maketitle
\begin{abstract}
  Strategy iteration is a technique frequently used for two-player
  games in order to determine the winner or compute payoffs, but to
  the best of our knowledge no general framework for strategy
  iteration has been considered. Inspired by previous work on simple
  stochastic games, we propose a general formalisation of strategy
  iteration for solving least fixpoint equations over a suitable class
  of complete lattices, based on MV-chains. We devise algorithms that
  can be used for non-expansive fixpoint functions represented as
  so-called min- respectively max-decompositions. Correspondingly, we
  develop two different techniques: strategy iteration from above,
  which has to solve the problem that iteration might reach a fixpoint
  that is not the least, and from below, which is algorithmically
  simpler, but requires a more involved correctness argument. We apply
  our method to solve energy games and compute behavioural metrics for
  probabilistic automata.
\end{abstract}

\section{Introduction}
\label{se:introduction}

Strategy iteration (or policy iteration) is a well known technique in computer
science. It has been widely adopted for the solution of two-player
games where the players, say {\Max} and {\Min}, aim at maximising and
minimising, respectively, some payoff.
In many cases there exists an optimal strategy for each player where
no deviation is advisable as long as the other player plays optimally.
We here assume a scenario where memoryless
(or positional) strategies are sufficient.
The general idea of strategy iteration is to iteratively fix a
strategy for one player, compute the optimal answering strategy for
the other player and then improve the strategy of the first player. As
long as there are only finitely many strategies, an optimal strategy
is bound to be found at some point. Such strategy iteration
  methods exist for Markov decision processes~\cite{How:DPMP} and for
  a variety of games, such as simple stochastic games \cite{condon92,hk:nonterminating-stochastic-games,ABS:GSIM},
  (discounted) mean-payoff games~\cite{ZP:CMPGG,BrimChal:2010} and parity
  games~\cite{VJ:DSIM,Sche:OSMPG}.

Similar ideas apply also to a wide range of different problems. For
instance, the computation of behavioural distances for systems
embodying quantitative information, e.g., time, probability or cost,
is often based on some form of lifting of distances on
states~\cite{bblm:on-the-fly-exact-journal,bbkk:coalgebraic-behavioral-metrics,bkp:up-to-behavioural-metrics-fibrations}. In
turn the lifting relies on couplings which play the role of strategies
and algorithms
% which are
based on a progressive improvement of couplings have been
devised~\cite{bblm:on-the-fly-exact-journal,bblmtb:prob-bisim-distance-automata-journal}.

\subparagraph*{Motivating example.}
To help the intuition, we
% briefly
review simple stochastic
games (SSGs)~\cite{condon92} and strategy iteration in that
setting as discussed in~\cite{bekp:fixpoint-theory-upside-down,DBLP:journals/corr/abs-2101-08184}.
% We will use SSGs as our running example.
An SSG consists of a set of states $V$, partitioned in four subsets
$\mathit{MIN}$, $\mathit{MAX}$, $\mathit{AV}$ and $\mathit{SINK}$.
States in $\mathit{SINK}$ (sink states) have no successor and yield a
payoff in $\interval{0}{1}$.  For states in $\mathit{AV}$ (average
states) the successor is determined by a probability distribution over
$V$, i.e., intuitively, the environment makes a probabilistic
choice. In a state in $\mathit{MIN}$, player {\Min} chooses a
successor trying to minimise the expected payoff, while in a state in
$\mathit{MAX}$, it is the player {\Max} that chooses, with the aim of
maximising the expected payoff. An example of an SSG is
% given
in Fig.~\ref{fi:ssg-example} on
page~\pageref{fi:ssg-example}.

When {\Min} and {\Max} play optimally, the expected payoff at each
state is given by the least fixpoint of the function
$\mathcal{V}\colon [0,1]^V\to [0,1]^V$, defined for
$a\colon V\to [0,1]$ and $v \in V$ by
\begin{align*}
  \mathcal{V}(a)(v) =
  \small
  \begin{cases}
    \max_{v\to v'} a(v') &v\in \mathit{MAX} \\
    \min_{v\to v'} a(v') &v\in \mathit{MIN} \\
    \sum_{v'\in V} p(v)(v')\cdot a(v') &v\in \mathit{AV} \\
    c(v) & v\in \mathit{SINK}
  \end{cases}
\end{align*}
with $p(v)(v')$
% denotes
the probability of state $v$ reaching
$v'$ and $c(v)\in [0,1]$ the payoff of sink state $v$.

The idea of strategy iteration from below, instantiated to this
context, is to compute the least fixpoint $\mu \mathcal{V}$ via an iteration of
the following kind:
\begin{enumerate}

\item Guess a strategy $\sigma \colon \mathit{MAX}\to V$ for player
  {\Max}, i.e., fix a successor for states in $\mathit{MAX}$.

\item Compute the least fixpoint of
  $\mathcal{V}_\sigma : [0,1]^V\to [0,1]^V$, which is defined
  as $\mathcal{V}$ in all cases apart from
  $v\in \mathit{MAX}$, where we set
  $\mathcal{V}(a)(v) = a(\sigma (v))$.  This fixpoint computation is
  simpler than the original one and it can be done efficiently via
  linear programming.

\item Based on $\mu \mathcal{V}_\sigma$, try to improve the strategy
  for {\Max}. If the strategy does not change, we have computed a
  fixpoint of $\mathcal{V}$ and, since iteration is from below, this is necessarily the least fixpoint. If the strategy changes, continue with
  step~2.
\end{enumerate}
 
A similar approach can be used for converging to the least fixpoint
from above. In this case, it is now player {\Min} who fixes a strategy
which is progressively improved.
This procedure is well-known to work for stopping
games~\cite{condon92}, i.e., SSGs where each combination of strategies
ensures termination, since for these games $\mathcal{V}$ has a unique
fixpoint.  However, in general, when iterating from above the procedure may get stuck at some
fixpoint which is not the least fixpoint of $\mathcal{V}$, a problem
which is solved by the theory developed
in~\cite{bekp:fixpoint-theory-upside-down} which can be used to
``skip'' this fixpoint and continue the iteration from there.

% With respect to value iteration, the algorithm above has the advantage
% of computing $\mu f = \mu f_{\sigma^*}$ whenever the optimal strategy
% $\sigma^*$ for player {\Min} has been found.

% Note that we require a finite number of strategies otherwise the
% above algorithm may never terminate.

% We also forward that {\Max} and {\Min} can swap roles in this algorithm,
% i.e.\ {\Max} updates his strategy in each iteration. These two differing
% strategy iterations have some advantages and disadvantages which we
% will discuss.

\medskip

While, as explained above, the general idea of strategy iteration is used in many different settings, to the best of our knowledge a general definition of strategy iteration is still missing.
%
% Taking inspiration from existing approaches,
The goal of the present
paper is to provide a general and abstract formulation of an algorithm
for strategy iteration, proved correct once and for all, which
instantiates to a variety of problems.
%
% For example, our generalized algorithm is capable of computing
% behavioural distances for labeled probabilistic transition
% systems. Bacci et al. \cite{bblm:on-the-fly-exact-journal} found a
% special algorithm which computes the behavioural distance of states
% in these systems. This motivated us to look for a generalization to
% this algorithm which can be classified as some sort of strategy
% iteration.
%
The key observation is that optimal strategies very often arise from
some form of extremal (least or greatest) fixpoint of a suitable
non-expansive function $f$ over a complete MV-chain~\cite{Mun:MV}, the paradigmatic example being the real interval $[0,1]$ with the usual order.
We propose a framework where the operation of fixing a strategy for
one of the players is captured abstractly, in terms of so-called min-
or max-decompositions of the function of interest. Then, we devise
strategy iteration approaches which converge to the fixpoint of
interest by successively improving the strategy for the chosen
player. We will assume that the interest is in least fixpoints, but
the theory can be dualised.  We propose two strategy iteration
algorithms that converge to the least fixpoint ``from below'' and
``from above'', respectively. As it happens for SSGs, in the latter
case the iteration can reach a fixpoint which is not the
least. Clearly, whenever the function $f$ of interest has a unique
fixpoint this problem disappears. Moreover, in some cases, even though
$f$ has multiple fixpoints, it can be ``patched'' in a way that the
modified function has the fixpoint of interest
% (least or greatest)
as
its only fixpoint.
%
% When instead $f$ does not have a unique fixpoint and it cannot be
% patched to have a single fixpoint (or the procedure that changes $f$
% causes efficiency problems),
% also in the general setting,
Otherwise, we can rely
on the results in~\cite{bekp:fixpoint-theory-upside-down} to check
whether the reached fixpoint is the least one and whenever it is not,
to get closer to the desired fixpoint and continue the
iteration.

% The developed theory in \cite{bekp:fixpoint-theory-upside-down}
% allows us to detect whether a fixpoint of $f$ is the desired extreme
% fixpoint. For this theory to work properly, the fixpoint function
% $f$ needs to be of a somewhat restricted special form. Now, if we
% get stuck at any non-extreme fixpoint, we can detect that it is not
% the desired fixpoint, reduce or increase its values, and continue
% the strategy iteration from there.

Strategy iteration approaches can be slow if compared to other
algorithms, such as value iteration. However, the benefit of strategy
iteration algorithms is that they allow an exact computation of the
desired fixpoint, while other algorithms may never reach the
sought-after extreme fixpoint but only converge towards it. This is
the case, e.g., for simple stochastic games, where strategy iteration
algorithms are the standard methods to obtain exact results.
Additionally, strategy iteration, besides determining
the fixpoint also singles out a strategy which allows one to obtain
it, an information which is often of interest.

In summary, we propose the first, to the best of our knowledge,
general definition of strategy iteration providing a lattice-theoretic
formalisation of this technique. This requires to single out and solve
in this general setting the fundamental challenges of these
approaches, which already show up in earlier work on SSGs (see,
e.g.,~\cite{bekp:fixpoint-theory-upside-down,BrimChal:2010}). In the
iteration from above, we may converge to a fixpoint that is not the
least, while from below it is not straightforward to show that
improving the strategy of {\Max} leads to a larger fixpoint.

Known algorithms are rediscovered for SSGs
and probabilistic automata~\cite{bblmtb:prob-bisim-distance-automata-journal}. Moreover new ones are obtained for energy
games~\cite{BFLMS:2008,CdAHS:2003} where movements in the game graph have an energy cost and the goal of one of the players is to avoid that the energy drops below zero.
Given the number of different application domains where
strategy iteration is or can be used, we feel that a general framework
can unveil unexplored potentials.
The two case studies (energy games and behavioural metrics) that we
treat can be encoded into SSGs~\cite{bblmtb:prob-bisim-distance-automata-journal}, but the obtained
strategies have to be translated back to the original setting, which
is not always trivial in general, and encodings usually come with a
loss of efficiency. For instance, in order
to solve SSGs a solver for linear programming is usually required,
which is in general not needed for other applications.

%\subparagraph*{Synopsis.} 
The rest of the paper is structured as follows. In
\S\ref{ss:lattices} we review some order-theoretic notions
% used in the paper
and recap some results
from~\cite{bekp:fixpoint-theory-upside-down} for identifying least
and greatest fixpoints.
In \S\ref{se:GSI} we devise two generalized strategy iteration
algorithms, from above and from below, using SSGs (already treated in
\cite{DBLP:journals/corr/abs-2101-08184}) as a running example. In
\S\ref{se:EG}, we show how our technique applies to energy games,
while in \S\ref{se:BM} we discuss an application to the computation of
the behavioural distance for probabilistic automata.  

The present paper is the full version of
\cite{bekp:lattice-strategy-iteration}.  Proofs and further material
can be found in the appendix.

\section{Preliminaries on ordered structures and fixpoints}
%\label{se:setup}
\label{ss:lattices}
\label{ss:upsidedown}

This section reviews some background used throughout the
paper. This includes the basics of lattices and MV-algebras,
where the functions of interest take values. We also recap some
results from~\cite{bekp:fixpoint-theory-upside-down} useful for detecting if a fixpoint of a given function is the least (or greatest).

% We will present a short overview of the developed theory. Details and proofs can be found in \cite{bekp:fixpoint-theory-upside-down}.

%\subsection{Complete lattices and MV-algebras}
%\label{ss:lattices}

For $X,Y$ sets, we denote by $\Pow{X}$ the powerset of $X$ and
$\Powf{X}$ the set of finite subsets of $X$. Moreover, the set
of functions from $X$ to $Y$ is denoted by either $Y^X$ or
$X\to Y$.

A partially ordered set $(P, \sqsubseteq)$
is often denoted simply as $P$, omitting the order
relation.
\pb{Given $x, y \in P$ we will denote by $\interval{x}{y}$ the interval $\{ z \in P \mid x \leq z \leq y \}$.}
For a function $f : X \to P$, we will write $\arg\min_{x\in X} f(x)$
to denote the set of elements where $f$ reaches the minimum, i.e.,
$\{ x \in X \mid \forall y \in X.\, f(x) \sqsubseteq f(y)\}$ and,
abusing the notation, we will write $z = \arg\min_{x\in X} f(x)$
instead of $z \in \arg\min_{x\in X} f(x)$.

The \emph{join} and the \emph{meet} of a
subset $X \subseteq P$ (if they exist) are denoted  $\bigsqcup X$
and $\bigsqcap X$.
 
A \emph{complete lattice} is a partially ordered set
$(\latL, \sqsubseteq)$ such that each subset $X \subseteq \latL$
admits a join $\bigsqcup X$ and a meet $\bigsqcap X$. A complete
lattice $(\latL, \sqsubseteq)$ always has a least element
$\bot = \bigsqcap \latL$ and a greatest element
$\top = \bigsqcup \latL$.

A function $f : \latL \to \latL$ is \emph{monotone} if for all
$l, l' \in \latL$, if $l \sqsubseteq l'$ then
$f(l) \sqsubseteq f(l')$. By Knaster-Tarski's
theorem~\cite[Theorem~1]{t:lattice-fixed-point}, any monotone
function on a complete lattice has a least fixpoint $\mu f$, characterised as the meet
of all pre-fixpoints
%\begin{center}
$\mu f = \bigsqcap \{ l \mid f(l) \sqsubseteq l \}$
% \qquad
and, dually, a greatest fixpoint
$\nu f = \bigsqcup \{ l \mid l \sqsubseteq f(l) \}$, characterised as
the join of all post-fixpoints.
We denote by $\mathit{Fix}(f)$ the set of all fixpoints of $f$.

Given a set $Y$ and a complete lattice $\latL$, the set of functions
$\latL^Y = \{ f \mid f : Y \to \latL \}$, endowed with pointwise
order, i.e., for $a, b \in \latL^Y$, $a \sqsubseteq b$ if
$a(y) \sqsubseteq b(y)$ for all $y\in Y$, is a complete lattice.
We write $a \sqsubset b$ when $a \sqsubseteq b$ and $a \neq b$, i.e., for
all $y \in Y$ we have $a(y) \sqsubseteq b(y)$ and
$a(y) \sqsubset b(y)$ for some $y \in Y$.

We are also interested in the set of probability distributions
$\mathcal{D}(Y)\subseteq [0,1]^Y$, i.e., functions
$\beta : Y \to [0,1]$ such that $\sum_{y \in Y} \beta(y) = 1$.

\medskip

%Next, we shortly present the theory developed in \cite{bekp:fixpoint-theory-upside-down} starting with  the required preliminaries.

%\paragraph*{MV-algebras.} 
%We first define MV-chains and depict the MV-chains exclusively used in this paper.

%\begin{definition}[MV-algebra]
%  \label{de:mv}
  An \emph{MV-algebra}~\cite{Mun:MV} is a tuple
  $\monM = (M, \oplus, \zero, \comp{(\cdot)})$ where
  $(M, \oplus, \zero)$ is a commutative monoid and
  $\comp{(\cdot)} : M \to M$ maps each element to its
  \emph{complement}, such that for all $x, y \in M$
  \begin{enumerate}
  \item \label{de:mv:1}
    $\comp{\comp{x}} = x$

  \item \label{de:mv:2}
    $x \oplus \comp{\zero} = \comp{\zero}$
    
  \item \label{de:mv:3}
    $\comp{(\comp{x} \oplus y)} \oplus y = \comp{(\comp{y} \oplus x)} \oplus x$.
  \end{enumerate}
  We denote $\one = \comp{\zero}$
  % , multiplication
  % $x \otimes y = \comp{\comp{x} \oplus \comp{y}}$
  and subtraction
  $x \ominus y = \comp{\comp{x} \oplus y}$.

  MV-algebras are endowed with a partial order, the so-called
\emph{natural order}, defined for $x,y\in M$, by $x \sqsubseteq y$
if $x \oplus z= y$ for some $z \in M$. When $\sqsubseteq$ is total,
$\monM$ is called an \emph{MV-chain}. We will write $\monM$ instead of
$M$.

The natural order gives an MV-algebra a lattice structure where
$\bot = \zero$, $\top =\one$, $x \sqcup y = (x \ominus y) \oplus y$
and
$x \sqcap y = \comp{\comp{x} \sqcup \comp{y}} = x \ominus (x \ominus y)$. We call the MV-algebra \emph{complete} if it is a
complete lattice, which is not true in general, e.g.,
$([0,1] \cap \mathbb{Q}, \leq)$.

\begin{example}
  \label{ex:mv-chains}
  A prototypical example of an MV-algebra is
  $([0,1],\oplus,0,\comp{(\cdot)})$ where $x\oplus y = \min\{x+y,1\}$,
  $\comp{x} = 1-x$ and $x\ominus y = \max\{0,x-y\}$ for $x,y\in [0,1]$. 
  The natural order is $\le$ (less
  or equal) on the reals.
  Another example is $K = (\{0,\dots,k\},\oplus,0,\comp{(\cdot)})$ where
  $n\oplus m = \min\{n+m,k\}$, $\comp{n} = k-n$ and $n\ominus m = \max\{ n-m,0\}$ for
  $n,m\in \{0,\dots,k\}$. Both MV-algebras are complete and MV-chains.
\end{example}

We next briefly recap the theory
from~\cite{bekp:fixpoint-theory-upside-down} which will be helpful in
the paper for checking whether a fixpoint is the least
or the greatest
fixpoint of some underlying endo-function.

\begin{remark}
  Hereafter, unless stated otherwise, $Y,Z$ will be assumed to be
  finite sets and $\monM$ will be a complete MV-chain.
\end{remark}

Given $a \in \monM^Y$ we define its \emph{norm} as
$\norm{a} = \max \{ a(y) \mid y \in Y\}$.  A function
$f: \monM^Y\to \monM^Z$ is \emph{non-expansive} if for all
$a, b \in \monM^Y$ it holds
$\norm{f(b) \ominus f(a)} \sqsubseteq \norm{b \ominus a}$. It can be
seen that non-expansive functions are monotone. A number of standard
operators are non-expansive (e.g., constants, reindexing, max and min
over a relation, average), and non-expansiveness is preserved by
composition and disjoint union
\pb{(see Appendix~\ref{se:decomposition-approximation-f} and~\cite{bekp:fixpoint-theory-upside-down})}. Given
$Y' \subseteq Y$ and $\delta \in \monM$, we write $\delta_{Y'}$ for
the function defined by $\delta_{Y'}(y) = \delta$ if $y \in Y'$ and
$\delta_{Y'}(y) = \zero$, otherwise.

For a non-expansive endo-function $f: \monM^Y\to \monM^Y$ and
$a \in \monM^Y$, the theory in~\cite{bekp:fixpoint-theory-upside-down}
provides  a so-called $a$-approximation $f_\#^a$ of $f$, which is an
endo-function over a suitable subset of $Y$.
More precisely, define $[Y]^a = \{ y\in Y\mid a(y) \neq 0 \}$ and
$\delta^a = \min \{ a(y)\mid y\in [Y]^a\}$. For $0\sqsubset\delta \in \monM$ consider the functions
$\alpha^{a,\delta} : \Pow{\Ytop{Y}{a}} \to \interval{a \ominus
  \delta}{a}$ and
$\gamma^{a,\delta} : \interval{a \ominus \delta}{a} \to
\Pow{\Ytop{Y}{a}}$, defined, for $Y' \in \Pow{\Ytop{Y}{a}}$ and
$b \in \interval{a \ominus \delta}{a}$, by
\begin{center}
  $
  \alpha^{a,\delta}(Y') = a \ominus \delta_{Y'} \qquad  \qquad
  \gamma^{a,\delta}(b) = \{ y\in \Ytop{Y}{a} \mid a(y) \ominus b(y)
  \sqsupseteq \delta \}.
  $
\end{center}
%Note that $\alpha$, $\gamma$ form a Galois connection.

For a non-expansive function $f: \monM^Y\to \monM^Z$ and
$\delta \in \monM$, define
$f^{a,\delta}_\# \colon \Pow{\Ytop{Y}{a}} \to \Pow{\Ytop{Z}{f(a)}}$
as
$f_\#^{a,\delta} = \gamma^{f(a),\delta} \circ f \circ
\alpha^{a,\delta}$. The function $f_\#^{a,\delta}$ is antitone in
the parameter $\delta$ and there exists a suitable value
$\inc{f}{a} \sqsupset 0$,
such that all
functions $f_\#^{a,\delta}$ for
$0\sqsubset \delta \sqsubseteq \inc{f}{a}$ are equal. The function
$f_\#^a := f_\#^{a,\inc{f}{a}}$ is called the
\emph{$a$-approximation} of $f$. When $\delta \sqsubseteq \delta_a$, the pair $\langle \alpha^{a,\delta}, \gamma^{a,\delta}\rangle$ is a Galois connection, a notion at the heart of abstract
interpretation~\cite{cc:ai-unified-lattice-model,CC:TLA}, and $f^{a,\delta}_\# = \gamma^{f(a),\delta} \circ f \circ \alpha^{a,\delta}$
is the best correct approximation of $f$.

Intuitively, given some $Y'$, the set $f_\#^a(Y')$ contains the points where a decrease of the values of $a$ on the points in $Y'$ ``propagates'' through the function $f$.
The greatest fixpoint of $f_\#^a$
gives us the subset of $Y$ where such a decrease is propagated in a
cycle (so-called ``vicious cycle''). Whenever $\nu f_\#^a$ is
non-empty, one can argue that $a$ cannot be the least fixpoint of $f$ since
we can decrease the value in all elements of $\nu f_\#^a$, obtaining a
smaller prefixpoint. Interestingly, for non-expansive functions, it is shown in~\cite{bekp:fixpoint-theory-upside-down} that also the converse holds, i.e., emptiness of the greatest fixpoint of $f_\#^a$ implies
that $a$ is the least fixpoint.

\begin{theorem}[soundness and completeness for fixpoints]
  \label{th:fixpoint-sound-compl}
  Let $\monM$ be a complete
  MV-chain, $Y$ a finite set and $f : \monM^Y\to \monM^Y$ be a
  non-expansive function.
  Let $a \in \monM^Y$ be a fixpoint of $f$. Then 
  $\nu f^a_{\#} = \emptyset$ if and only if $a = \mu f$.
\end{theorem}

Using the above theorem we can check whether some fixpoint $a$ of $f$
is the least fixpoint. Whenever $a$ is a fixpoint, but not yet the
least fixpoint of $f$, it can be decreased
by a fixed value in $\monM$
% $\inc{f}{a} \in \monM$
(see \cite{bekp:fixpoint-theory-upside-down} for the details)
on the points in $\nu f^a_{\#}$
to obtain a smaller pre-fixpoint.

\begin{lemma}
  \label{lem:fp-increase}
  Let $\monM$ be a complete MV-chain, $f : \monM^Y\to \monM^Y$ a
  non-expansive function, $a \in \monM^Y$ a fixpoint of $f$, and let
  $f^a_{\#}$ be the corresponding $a$-approximation.
  % and $\inc{a}{f}$
  %as in Def.~\ref{de:approximation-2}.
  If $a$ is not the least fixpoint and thus
  $\nu f^a_{\#} \neq \emptyset$ then there is
  $0 \sqsubset \delta \in \monM$ such that
  $a \ominus \delta_{\nu f^a_{\#}}$ is a pre-fixpoint of $f$.
\end{lemma}

In the following we will use this result as a ``black box'': we assume
that given $f$ and a fixpoint $a$ of $f$ we can determine whether
$a=\mu f$ and, if not, obtain $a' \sqsubset a$ such
$f(a') \sqsubseteq a'$.

The above theory can easily be dualised (see
\cite{bekp:fixpoint-theory-upside-down} for the details of the dual view).

\section{Generalized strategy iteration}
\label{se:GSI}

In this section we develop two strategy iteration techniques for
determining least fixpoints. The first technique requires a so-called
min-decomposition and approaches the least fixpoint from above, while
the second uses a max-decomposition to ascend to the least fixpoint from
below.

Hence fixpoint iteration from above is seen strictly from the point of
view of the {\Min} player, while fixpoint iteration from below is from
the view of the {\Max} player, who want to minimize respectively maximize
the payoff. The player starts by guessing a strategy, which in the
case of the {\Min} (\Max) player over-approximates (under-approximates)
the true payoff. This strategy is then locally improved at each
iteration based on the payoff produced by the player following such a
strategy. That is, we compute fixpoints for a fixed strategy, which in
a two-player game means that the opponent plays optimally. When the
set of strategies is finite (or, at least, the search can be
restricted to a finite set), an optimal strategy will be found at some
point.

\subsection{Function decomposition}

We next introduce the setting where the generalisations of strategy
iteration will be developed. We assume that the game we are interested
in is played on a finite set of positions $Y$ and the payoff at each
position is an element of a suitable complete MV-chain $\monM$. This
payoff is given by a function in $\monM^Y$ that can be characterised
as the least fixpoint of a monotone function
$f : \monM^Y \to \monM^Y$. If we concentrate on the {\Min} player, each position $y\in Y$ is assigned a set of functions
$H_{\min}(y)\subseteq (\monM^Y \to \monM)$ where each function
$h\in H_{\min}(y)$ is one possible option that can be chosen by
{\Min}. Given $a\colon Y\to\monM$ as the current estimate of the
payoff, $h(a)$ is the resulting payoff at $y$. If the player does not
have a choice, this set is a singleton. Since it is the aim of {\Min}
to minimise she will choose an $h$ such that $h(a)$ is
minimal.

\begin{definition}[min-decomposition]
  \label{de:min-decomposition}
  Let $Y$ be a finite set and $\monM$ be a complete MV-chain. Given a function
  $f : \monM^Y \to \monM^Y$, a \emph{min-decomposition} of $f$ is a
  function $H_{\min} : Y \to \Powf{\monM^Y \to \monM}$ such that for
  all $y\in Y$ the set $H_{\min}(y)$ consists only of monotone
  functions and for all $a \in \monM^Y$ it holds
  %\begin{center}
    $f(a)(y) = \min_{h \in H_{\min}(y)}  h(a)$.
  %\end{center}
\end{definition}

Observe that any monotone function $f : \monM^Y \to \monM^Y$ admits a
trivial min-decomposition $I$ defined by $I(y) = \{ h_y \}$ where
$h_y(a) = f(a)(y)$ for all $a \in \monM^Y$.

Whenever all $h\in H_{\min}(y)$ are not only monotone, but also
non-expansive, it can be shown easily that $f$ is also non-expansive
and we can obtain an approximation as discussed in
\S\ref{ss:upsidedown}. How this can be done compositionally is
explained in Appendix~\ref{se:decomposition-approximation-f}.
Max-decompositions, with analogous properties, are defined dually, i.e.\
$H_{\max}: Y \to \Powf{\monM^Y \to \monM}$ and
%\begin{center}
  $f(a)(y) = \max_{h \in H_{\max}(y)}  h(a)$.
%\end{center}

\begin{toappendix}
  \subsection{Decomposition of $f$ and derivation of the approximation
    $f^a_\#$}
  \label{se:decomposition-approximation-f}

  In this section we show how a min-decomposition (analogously a
  max-decomposition) of a mapping $f\colon \monM^Y \to \monM^Y$ can be
  assembled using the basic functions and operators introduced in
  \cite{bekp:fixpoint-theory-upside-down}. This serves two purposes:
  in this way we show that $f$ is automatically non-expansive if
  obtained from non-expansive components. Second, this gives us a
  recipe to obtain the approximation $f^a_\#$, required for checking
  whether a given fixpoint is indeed the least.

  Table~\ref{tab:basic-functions-approximations} lists basic
  non-expansive functions and operators for composing them. Note that
  all those functions are non-expansive and the operators preserve
  non-expansiveness. In addition the table lists the corresponding
  approximations.

  \begin{table}[h]
    \normalsize
    \caption{\normalsize Basic functions $f\colon \monM^Y\to \monM^Z$
      (constant, reindexing, minimum, maximum, average), function
      composition, disjoint union and the corresponding approximations
      $f^a_\#\colon \Pow{\Ytop{Y}{a}} \to
      \Pow{\Ytop{Z}{f(a)}}$. 
    }
    \label{tab:basic-functions-approximations}
    \small
    \smallskip
    \emph{Notation:} $\mathcal{R}^{-1}(z) = \{y\in Y \mid
    y\mathcal{R}z\}$, $\mathit{supp}(p) = \{y\in Y\mid p(y) > 0\}$ for
    $p\in \mathcal{D}(Y)$, \\
    $\arg\min_{y\in Y'} a(y)$, resp.\ $\arg\max_{y\in Y'} a(y)$,
    the set of elements where $a|_{Y'}$ \\
    reaches the minimum, resp.\ the maximum,
    for $Y' \subseteq Y$ and $a\in\monM^Y$
    \begin{center}
      \begin{tabular}{|l|l|l|}
        \hline function $f$ & definition of $f$ & $f_\#^a(Y')$ (below) \\
        \hline\hline
        $c_k$ & $f(a) = k$ & $\emptyset$ \\
        ($k\in\monM^Z$) && \\ \hline
        $u^*$ & $f(a) = a \circ u$ & $u^{-1}(Y')$ \\
        ($u\colon Z\to Y$) & & \\ \hline
        $\mins_\mathcal{R}$ &
        $f(a)(z) = \min\limits_{y\mathcal{R}z} a(y)$ &
        $\{z \in \Ytop{Z}{f(a)} \mid
        \arg\min\limits_{y\in\mathcal{R}^{-1}(z)} a(y)\cap Y' \neq
        \emptyset\}$\mystrutbl \\
        ($\mathcal{R}\subseteq Y\times Z$) & &
        \\ \hline
        $\maxs_\mathcal{R}$ &
        $f(a)(z) = \max\limits_{y\mathcal{R}z} a(y)$ &
        $\{z \in \Ytop{Z}{f(a)} \mid
      \arg\max\limits_{y\in\mathcal{R}^{-1}(z)} a(y) \subseteq Y'\}$
        \mystrutbl \\
        ($\mathcal{R}\subseteq Y\times Z$) &&
        \\
        \hline \mystrutab$\mathrm{av}_D$ \quad ($\monM = [0,1]$, &
        $f(a)(p) = \sum\limits_{y\in Y} p(y)\cdot a(y)$ &
        $\{p \in \Ytop{D}{f(a)} \mid \mathit{supp}(p) \subseteq Y'\}$ \\
        $Z = D \subseteq \mathcal{D}(Y)$) & &
        \\
        \hline $h\circ g$ & $f(a) = h(g(a))$ &
        $h^{g(a)}_\# \circ g^a_\#(Y')$  \\
        ($g\colon \monM^Y\to \monM^W$, && \\
        $h\colon \monM^W\to \monM^Z$) && \\ \hline
        $\biguplus\limits_{i\in I} f_i$ \quad $I$ finite &
        $f(a)(z) = f_i(a|_{Y_i})(z)$ &
        $\biguplus_{i\in I} (f_i)^{a|_{Y_i}}_\#(Y'\cap Y_i)$ 
        \\
        ($f_i\colon \monM^{Y_i}\to \monM^{Z_i}$, &
        ($z\in Z_i$) &
        \\
        $Y = \bigcup\limits_{i\in I} Y_i$,
        $Z = \biguplus\limits_{i\in I} Z_i$)\mystrutbl & & \\
        \hline
      \end{tabular}
    \end{center}
  \end{table}

  We will now show how to obtain the approximation of a function $f$
  given its min-decomposition and approximations for all the functions
  used in the min-decomposition.

  \begin{proposition}
    Let $Y$ be a finite set and $\monM$ a complete MV-chain. Let
    $f\colon \monM^Y \to \monM^Y$ be a function and
    $H_{\min} \colon Y \to \Powf{\monM^Y \to \monM}$ a given
    min-decomposition such that, for all $y\in Y$, all functions
    $h\in H_{\min}(y)$ are non-expansive. Then $f$ is non-expansive
    and the approximation
    $f_\#^a(Y') \colon \Pow{[Y]^a} \to \Pow{[Y]^{f(a)}}$ is given by
    \[ f_\#^a (Y') = \{ y\in [Y]^{f(a)} \mid \exists h\,
      \big( h = {\arg\min}_{h'\in H_{\min}(y)} h'(a) \ \land\
      h_\#^a(Y') \neq \emptyset \big) \} \]
    for $a\in \monM^Y$ and
    $Y'\subseteq [Y]^a$.
  \end{proposition}

\begin{proof}
  We first show that $f$ is non-expansive by proving that it can be
  expressed as a composition of non-expansive functions. (Recall that
  function composition and disjoint union preserve non-expansiveness.)

  For all $y \in Y$, let $H_{\min}(y)= \{ h_y^1,\dots ,h_y^{k_y}\}$ and
  let $I_y = \{ 1,\dots , k_y\}$ be the corresponding index set.
  We have
  $h_y^i \colon \monM^Y \to \monM$ for each $i\in I_y$, where -- for
  convenience -- we view each function $h_y^i$ as being of type
  $\monM^Y \to \monM^{\{i\}}$, where $\{i\}$ is the singleton set
  containing $i$.

  We introduce auxiliary functions $g_y\colon \monM^Y\to \monM^{I_y}$
  and $g\colon \monM^Y \to \monM^I$ (with $I = \biguplus_{y\in Y}I_y$)
  defined as below, where  $a\in \monM^Y$ and
  $i\in I_y$:  
  \[ g_y = \biguplus_{j\in I_y} h_y^j \quad g_y(a)(i) =
    h_y^{i}(a) \qquad g= \biguplus_{y\in
      Y}g_y\quad g(a)(i) = g_y(a)(i) \]
  Next we define $u\colon I \to Y$
  where $u(i) = y$ for all $i\in I_y$. This allows decomposition
  of $f$ as
  \[ f = \min\nolimits_u \circ g \] with
  $\min_u \colon \monM^I\to \monM^Y$. More
  intuitively, given $a\in \monM^Y$ and $y\in Y$, we have
  \[ f(a)(y) = \min_{i,u(i) = y} g_y(a)(i) = \min_{i\in I_y}
    h^i_y(a)\] and
  \[ f(a) = \min\nolimits_u(g(a)) = \min\nolimits_u
    (\biguplus_{y\in Y}g_y(a)) = \min\nolimits_u (\biguplus_{y\in
      Y}\biguplus_{u(i)=y} h_y^i (a)).\]

  Next, we want to express the approximation
  $f^a_\#\colon \Pow{[Y]^a} \to \Pow{[Y]^{f(a)}}$ for
  $a\in \monM^Y$ in terms of the approximations of the components
  $(h_y^{i})_\#^a \colon \Pow{[Y]^a} \to
  \Pow{[\{i\}]^{h_y^{i}(a)}}$.

  In the following $Y'\subseteq [Y]^a$ is some subset
  of $[Y]^a$.
  
  First observe, that 
  $(g_y)_\#^a \colon \Pow{[Y]^a} \to \Pow{[I_y]^{g_y(a)}}$, as recalled in Table~\ref{tab:basic-functions-approximations}, is given by
  \[ (g_y)_\#^a(Y') = \bigcup_{i\in I_y} \biggl( (h_y^i)^a_\# (Y')
    \biggr). \]
  Moreover
  $g^a_\# \colon \Pow{[Y]} \to \Pow{[I]^{g(a)}}$ is given by
  \[ g^a_\# (Y') = \bigcup_{y\in Y} (g_y)_\#^a = \bigcup_{y\in Y} \bigcup_{i\in I_y} \biggl(
    (h_y^i)^a_\# (Y') \biggr) \]
  Hence we obtain for $i\in I_y$:
  \[ i \in g_\#^a(Y') \text{ iff } i \in (g_y)_\#^a(Y') \text{ iff } i \in (h_y^{i})_\#^a (Y') \text{ iff } (h_y^{i})_\#^a (Y')\neq \emptyset.\]
  Finally, we can conclude
     \begin{align*}
       f_\#^a(Y') &= \{ y \in [Y]^{f(a)} \mid \mathit{Min}_{g(a)\mid u^{-1}(y)} \cap g^a_\#(Y') \neq \emptyset \} \\
       &= \{ y\in [Y]^{f(a)} \mid \exists i\, \big( i = {\arg\min}_{j \in I_y} g(a)(j) \ \land\  i \in g_\#^a(Y') \big) \} \\
       &= \{ y\in [Y]^{f(a)} \mid \exists i  \big( i={\arg\min}_{j\in I_y} g_y(a)(j) \ \land\  i \in (g_y)_\#^a(Y') \big) \} \\
       &= \{ y\in [Y]^{f(a)} \mid \exists i\, \big( i= {\arg\min}_{j\in I_y} h_y^{j}(a) \ \land\  i \in (h_y^{i})_\#^a(Y') \big) \} \\
       &= \{ y\in [Y]^{f(a)} \mid \exists h\, \big( h = {\arg\min}_{h'\in H_{\min}(y)} h'(a) \ \land\ 
       h_\#^a(Y')\neq \emptyset \big) \}
     \end{align*}
     %
%      \begin{align*}
%        f_\#^a(Y') &= \{ y \in [Y]^{f(a)} \mid \mathit{Min}_{g(a)\mid u^{-1}(y)} \cap g^a_\#(Y') \neq \emptyset \} \\
%        &= \{ y\in [Y]^{f(a)} \mid \exists i \in I_y \big( g(a)(i) = \min_{j_y\in I_y} g(a)(j_y) \ \land\  i \in g_\#^a(Y') \big) \} \\
%        &= \{ y\in [Y]^{f(a)} \mid \exists i \in I_y \big( g_y(a)(i) = \min_{j_y\in I_y} g_y(a)(j_y) \ \land\  i \in (g_y)_\#^a(Y') \big) \} \\
%        &= \{ y\in [Y]^{f(a)} \mid \exists i \in I_y \big( h_y^{i}(a) = \min_{j_y\in I_y} h_y^{j_y} \ \land\  i \in (h_y^{i})_\#^a(Y') \big) \} \\
%        &= \{ y\in [Y]^{f(a)} \mid \exists h\in H_{\min}(y) \big( h(a) = \min_{h'\in H_{\min}(y)} h'(a) \ \land\ 
%        h_\#^a(Y')\neq \emptyset \big) \}
% \end{align*}
  \end{proof}
  
  In a similar fashion, for a max-decomposition $H_{\max}$ of $f$, the
  approximation \linebreak $f_\#^a(Y') \colon \Pow{[Y]^a} \to \Pow{[Y]^{f(a)}}$ is given by
  \[ f_\#^a (Y') = \{ y\in [Y]^{f(a)} \mid \forall h\,
    \big(h = {\arg\max}_{h'\in H_{\max}(y)} h'(a)\ \Rightarrow \
    h_\#^a(Y') \neq \emptyset\big) \}
  \]
  for $a\in \monM^Y$ and
  $Y'\subseteq [Y]^a$.

\subsection{Proofs}
\label{ap:proofs-se-3}

\end{toappendix}

Fixing a strategy can be seen as fixing, for all $y\in Y$, some
element in $H_{\min}(y)$.

\begin{definition}[strategy]
  \label{de:strategy}
  Let $Y$ be a finite set, $\monM$ be a complete MV-chain,
  $f : \monM^Y \to \monM^Y$ and let
  $H_{\min} : Y \to \Powf{\monM^Y \to \monM}$ be a min-decomposition
  of $f$. A \emph{strategy} in $H_{\min}$ is a function
  $C : Y \to (\monM^Y \to \monM)$ such that for all $y \in Y$ it holds
  that $C(y) \in H_{\min}(y)$.
  For a fixed $C$ we define $f_C : \monM^Y \to \monM^Y$ as
  $f_C(a)(y) = C(y)(a)$ for all $a \in \monM^Y$ and $y \in Y$.

  Strategies in a max-decomposition are defined dually.
\end{definition}

The letter $C$ stands for ``choice'' and typically $\mu f_C$ is easier
to compute than $\mu f$.

\begin{example}
  \label{ex:sg}
  As a running example for illustrating our theory and the
  resulting algorithms
  we will use \emph{simple stochastic
    games} (SSGs). We first show that they fall into the framework. Fix
  an SSG with a finite set $V$ of states, partitioned into
  $\mathit{MIN}$, $\mathit{MAX}$, $\mathit{AV}$ (average) and
  $\mathit{SINK}$. Successors of $\mathit{MIN}$ and $\mathit{MAX}$
  states are given by a relation
  $\to\ \subseteq (\mathit{MIN}\cup\mathit{MAX}) \times V$, while
  $p : \mathit{AV} \to \interval{0}{1}^V$ maps each $v\in \mathit{AV}$
  to a distribution $p(v)\in \mathcal{D}(V)$.  Finally,
  $c : \mathit{SINK}\to [0,1]$ provides the payoff of sink states.

  The fixpoint function $\mathcal{V} : [0,1]^V\to [0,1]^V$, as defined
  in the introduction, admits a min-decomposition
  $H_{\min}\colon V \to \Powf{\monM^V\to \monM}$ defined for all
  $a \in \monM^Y$ as follows:
  \begin{itemize}
    
  \item for $v \in \mathit{MIN}$,
    $H_{\min}(v) = \{ h_{v'} \mid v \to v'\}$ with $h_{v'}(a) = a(v')$;
    
  \item for $v \in \mathit{MAX}$,
    $H_{\min}(v) = \{ h \}$ with $h(a) = \max_{v\to v'} a(v')$;
    
  \item for $v \in \mathit{AV}$,
    $H_{\min}(v) = \{ h \}$ with $h(a) = \sum_{v'\in V} p(v)(v')\cdot a(v')$;

  \item for $v \in \mathit{SINK}$,
    $H_{\min}(v) = \{ h \}$ with $h(a) = c(v)$.
    
  \end{itemize}
  A max-decomposition can be defined dually.

  \smallskip
  
  For instance, consider the SSG in Fig.~\ref{fi:ssg-example}
  where
  $V=\{ \state{1},\state{\varepsilon},\state{av},\state{max},\state{min} \}$
  with the obvious partitioning. The fixpoint function is
  $\mathcal{V}\colon [0,1]^V\to [0,1]^V$ defined, for $a\in [0,1]^V$,
  by
  \begin{eqnarray*}
    \mathcal{V}(a)(\state{1}) = 1 \qquad
    \mathcal{V}(a)(\state{\varepsilon})= \varepsilon \qquad
    \mathcal{V}(a)(\state{av}) = \frac{1}{2} a(\state{min}) + \frac{1}{2}
    a(\state{max}) \\
     \mathcal{V}(a)(\state{max}) = \max \{
    a(\state{\varepsilon}),a(\state{av})\} \qquad \mathcal{V}(a)(\state{min}) =
    \min \{ a(\state{1}),a(\state{av})\}.
  \end{eqnarray*}
  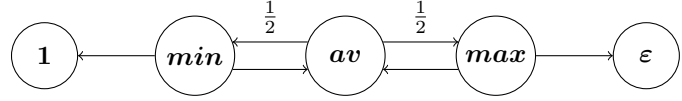
\begin{figure}[t]
    \begin{center}
      \begin{tikzpicture}
        \node (S1) at (0,0) [circle,draw,minimum width=2.5em] {${\state{1}}$};
        \node (MIN) at (2,0) [circle,draw,minimum width=3em,inner sep=0] {$\state{min}$};
        \node (AV) at (4,0) [circle,draw,minimum width=3em,inner sep=0] {$\state{av}$};
        \node (MAX) at (6,0) [circle,draw,minimum width=3em,inner sep=0] {$\state{max}$};
        \node (S0) at (8,0) [circle,draw,minimum width=2.5em] {$\state{\varepsilon}$};

        \draw[->] (MIN) to (S1);
        \draw[->] (MAX) to (S0);
        \draw (MIN.340) edge[auto=right,->] (AV.200);
        \draw (MAX.200) edge[auto=right,->] (AV.340);
        \draw (AV.20) edge[auto=right,->] node [above] {$\frac{1}{2}$} (MAX.160);
        \draw (AV.160) edge[auto=right,->] node [above] {$\frac{1}{2}$} (MIN.20);
      \end{tikzpicture}
    \end{center}
    \caption{An example of a simple stochastic game. States
      $\state{1}$, $\state{\varepsilon}$ have payoff $1$, $\varepsilon > 0$
      respectively.}
    \label{fi:ssg-example}
  \end{figure}
  The min-decomposition defined in general above, in this case is
  $H_{\min}\colon V \to \Powf{\monM^V\to \monM}$ defined for all $a \in \monM^Y$ as follows:
  for $v\in V\setminus \{ \state{min} \}$, we let $H_{\min}(v) = \{ h \}$
  with $h(a) = \mathcal{V}(a)(v)$, while
  $H_{\min}(\state{min}) = \{ h_{\state{1}}, h_{\state{av}}\}$
  with $h_{\state{1}}(a)=a(\state{1})$ and
  $h_{\state{av}}(a)= a(\state{av})$.
  All strategies in $H_{\min}$ assign to every state $v\in V\setminus \{ \state{min} \}$ the only element in $H_{\min}(v)$. Hence they
  are determined by the value on state $\state{min}$: thus there are two strategies $C_1^{\min}, C_2^{\min}$ in $H_{\min}$ with
  $C_1^{\min}(\state{min}) = h_{\state{1}}$ and
  $C_2^{\min}(\state{min}) = h_{\state{av}}$.

  Dually, a max-decomposition
  $H_{\max}\colon V \to \Powf{\monM^V\to \monM}$ is defined for all $a \in \monM^Y$ as follows:
  $H_{\max}(v) = \{ h \}$ with $h(a) = \mathcal{V}(a)(v)$ for all
  $v\in V\setminus \{ \state{max} \}$ and
  $H_{\max}(\state{max}) = \{ h_{\state{\varepsilon}},h_{\state{av}}\}$
  with $h_{\state{\varepsilon}}(a)=a({\state{\varepsilon}})$ and
  $h_{\state{av}}(a)= a(\state{av})$.
  Again, there are two strategies $C_1^{\max}$ and $C_2^{\max}$ in $H_{\max}$
  that differ for the value assigned to $\state{max}$: 
$C_1^{\max}(\state{max}) = h_{\state{\varepsilon}}$ and
  $C_2^{\max}(\state{max}) = h_{\state{av}}$.
\end{example}

\begin{toappendix}
The following lemma reports two easy observations which will be used
several times.

\begin{lemma}
  \label{le:min-set-basic}
  Let $Y$ be a finite set, $\monM$ be a complete MV-chain and let
  $H_{\min} : Y \to \Powf{\monM^Y \to \monM}$ be a min-decomposition of $f\colon \monM^Y\to \monM^Y$.

  \begin{enumerate}
  \item
    \label{le:min-set-basic:a}
    For all strategies $C$ in $H_{\min}$ we have that $f \sqsubseteq f_C$, pointwise.
  \item
    \label{le:min-set-basic:b}
    For all $a \in \monM^Y$, there is a strategy $C_a$ such that
    $C_a(y)(a) = f(a)(y)$ for all $y \in Y$.
  \end{enumerate}
\end{lemma}

\begin{proof}
  (\ref{le:min-set-basic:a}) Just note that for all $a \in \monM^Y$ and
  $y \in Y$, we have
  \begin{align*}
    f_C(a)(y) & = C(y)(a)\\
              & \sqsupseteq \min_{h \in H_{\min}(y)} h(a)
              & \mbox{[since $C(y) \in H_{\min}(y)$]}\\
              & = f(a)(y)
              & \mbox{[by definition of min-decomposition]}
  \end{align*}
  (\ref{le:min-set-basic:b}) For all $y \in Y$, it holds
  $f(a)(y) = \min_{h \in H_{\min}(y)} h(a) = h_y(a)$ for some
  $h_y \in H_{\min}(y)$ since the minimum is realised ($H_{\min}$ is 
  finite). And thus we can define $C_a(y) = h_y$.
\end{proof}

The above result can be easily adapted for a max-decomposition by
reversing the order.
\end{toappendix}

\subsection{Strategy iteration from above}
\label{SIfromBelow}

In this section we propose a generalized strategy iteration algorithm
from above. It is based on a min-decomposition of the function and,
intuitively, at each iteration the player {\Min} improves her
strategy. An issue here is that this iteration may get stuck at a
fixpoint strictly larger than the least one. Recognising and
overcoming this problem, thus continuing the iteration until the least
fixpoint is reached, requires the theory described in
\S\ref{ss:upsidedown}.

The basic result that motivates strategy iteration from above is a
characterisation of the least fixpoint of a function in terms of a
min-decomposition.

\begin{propositionrep}[least fixpoint from min-decompositions]
  \label{pr:above-muf}
  Let $Y$ be a finite set, $\monM$ a complete MV-chain,
  $f : \monM^Y \to \monM^Y$ a monotone function and let
  $H_{\min} : Y \to \Powf{\monM^Y \to \monM}$ be a min-decomposition of $f$.
  Then $\mu f = \min \{ \mu f_C \mid C \text{ is a strategy in } H_{\min}\}$.
\end{propositionrep}

\begin{proof}
  By Lemma~\ref{le:min-set-basic}(\ref{le:min-set-basic:a}), for all
  strategies $C$ in $H_{\min}$ we have $f \sqsubseteq f_C$, whence
  $\mu f \sqsubseteq \mu f_C$. Then
  $\mu f \sqsubseteq \min \{ \mu f_C \mid C \text{ is a strategy in } H_{\min}\}$ follows.

  For the converse inequality, note that by
  Lemma~\ref{le:min-set-basic}(\ref{le:min-set-basic:b}) there exists
  some strategy $C$ in $H_{\min}$, such that
  $f_C(\mu f)(y) = C(y)(\mu f) = f(\mu f)(y) = \mu f(y)$ for all
  $y\in Y$. Thus $\mu f$ is a (pre-)fixpoint of $f_C$ and thus
  $\mu f_c \sqsubseteq \mu f$. Hence
  $\mu f \sqsupseteq \min \{ \mu f_C \mid C \text{ is a strategy in } H_{\min}\}$.
\end{proof}

Although we do not focus on complexity issues, we observe that -- under suitable assumptions -- we can show that given
a function $f$ as a min-decomposition, the problem of checking whether
$\mu f \sqsubseteq b$ for some bound $b\in\monM^Y$ is in
$\mathsf{NP}$. For each $y\in Y$
% the size of the set $H_{\min}(y)$ of
% functions is bounded by the input size and
we can nondeterministically guess $C(y)\in H_{\min}(y)$ thus defining
a strategy.  Assuming that the computation of $\mu f_C$ is polynomial,
we can thus determine in non-deterministic polynomial time (in the
size of the representation of $f$) whether
$\mu f\sqsubseteq \mu f_C\sqsubseteq b$.

Now in order to compute the least fixpoint, the idea is to start from
some (arbitrary) strategy, say $C_0$, in $H_{\min}$. At each
iteration, if the current strategy is $C_i$ one tries to construct, on
the basis of $\mu f_{C_i}$, a new strategy $C_{i+1}$ which improves
$C_i$, in the sense that $\mu f_{C_{i+1}}$ becomes smaller. This
motivates the notion of improvement.

\begin{definition}[min-improvement]
  \label{de:improvement}
  Let $f : \monM^Y \to \monM^Y$ be a monotone function, where $Y$ is a
  finite set and $\monM$ a complete MV-chain, and let $H_{\min}$ be a
  min-decomposition. Given  strategies $C, C'$ in $H_{\min}$, we say
  that $C'$ is a \emph{min-improvement} of $C$ if
  $f_{C'}(\mu f_C) \sqsubset \mu f_C$.  It is called a \emph{stable}
  min-improvement if in addition $C'(y) = C(y)$ for all $y \in Y$ such
  that $f_{C'}(\mu f_C)(y) = \mu f_C(y)$. We denote by $\minimp{C}$
  (respectively $\minimps{C}$) the set of (stable) min-improvements of
  $C$.
\end{definition}

The notion of stability
% introduced here
will
% actually
turn out to be useful later,
for performing strategy iteration from below (as explained in the next section).
In a stable min-improvement, the player is only allowed to switch
the strategy in a state if this yields a strictly better payoff.
Interestingly, instances of this notion are adopted, more or less implicitly, in other strategy improvement algorithms in the literature (cf.\ \cite[Definition~13]{ABS:GSIM} and the way in which improvements are computed in \cite{BrimChal:2010}).
Clearly $\minimps{C} \subseteq \minimp{C}$.
In addition, it can be easily seen that there exists a stable min-improvement
as long as there is any improvement.

\begin{remark}[obtaining min-improvements]
  \label{rem:min-imp}
  For a strategy $C$, if $\minimp{C} \neq \emptyset$, one can obtain a
  min-improvement of $C$ by taking $C' \neq C$ defined as
  $C'(y) =\arg\min_{h\in H_{\min}(y)} h(\mu f_C)$
  and a stable
  min-improvement as:
  \begin{align*}
    C'(y) = \left\{
      \begin{array}{lll}
        C(y) & & \mbox{if $f(\mu f_C)(y) = \mu f_C(y)$}\\
        \arg\min_{h\in H_{\min}(y)} h(\mu f_C) & & \mbox{otherwise}
      \end{array}
    \right.
  \end{align*}
 There could be several $h\in H_{\min}(y)$ where $h(\mu f_C)$
is minimal. Any such choice is valid.
\end{remark}

We next show that, as suggested by the terminology, a min-improvement
leads to a smaller least fixpoint.

\begin{lemmarep}[min-improvements reduce fixpoints] 
  \label{le:improve-above}
  Let $Y$ be a finite set, $\monM$ a complete MV-chain,
  $f : \monM^Y \to \monM^Y$ a monotone function and $H_{\min}$ a
  min-decomposition of $f$. Given a strategy $C$ in $H_{\min}$ and a
  min-improvement $C' \in \minimp{C}$ it holds
  $\mu f_{C'} \sqsubset \mu f_{C}$.
\end{lemmarep}

\begin{proof}
  By definition of improvement, we have that
  $f_{C'}(\mu f_C) \sqsubset \mu f_C$, i.e., $\mu f_C$ is a
  pre-fixpoint of $f_{C'}$ and it is not a fixpoint. Hence by
  Knaster-Tarski $\mu f_{C'} \sqsubset \mu f_{C}$ follows.
\end{proof}

Thus, once the strategy can be improved, we will get closer to the least
fixpoint of $f$. We next show that an improvement of the current
strategy exists as long as we have not encountered a fixpoint of $f$.

\begin{lemmarep}[min-improvements exist for non-fixpoints]
  \label{lem:iff-above}
  Let $Y$ be a finite set, $\monM$ a complete MV-chain, $f : \monM^Y \to \monM^Y$ a
  monotone function and $H_{\min}$ a min-decomposition. Given a strategy
  $C$ in $H_{\min}$, the following are equivalent:
  \begin{enumerate}
  \item $\mu f_C \not\in \mathit{Fix}(f)$
  \item $\minimp{C} \neq \emptyset$
  \item $f(\mu f_C) \sqsubset \mu f_C$
  \end{enumerate}
\end{lemmarep}

\begin{proof}
  (1 $\Rightarrow$ 2): Assume $\mu f_C \notin \mathit{Fix}(f)$. By
  Lemma~\ref{le:min-set-basic}(\ref{le:min-set-basic:a}) we know
  $f \sqsubseteq f_C$. Hence
  $\mu f_C = f_C (\mu f_C) \sqsupseteq f (\mu f_C)$. Since
  $\mu f_C \notin \mathit{Fix}(f)$, we deduce that the inequality is
  strict $\mu f_C \sqsupset f (\mu f_C)$.

  By Lemma~\ref{le:min-set-basic}(\ref{le:min-set-basic:b}) we can
  take a strategy $C'$, s.t. for all $y\in Y$
  \begin{align*}
    f(\mu f_C)(y) = C'(y) (\mu f_C) = f_{C'}(\mu f_C)(y)
  \end{align*}
  Thus $\mu f_C \sqsupset f_{C'}(\mu f_C)$ and a min-improvement
  exists by definition.

  \medskip
  
  (2 $\Rightarrow$ 1): Let $C' \in \minimp{C}$. By definition of
  improvement $f_{C'}(\mu f_C)\sqsubset \mu f_C$. By
  Lemma~\ref{le:min-set-basic}(\ref{le:min-set-basic:a}) we know
  $f \sqsubseteq f_{C'}$ and thus
  $f(\mu f_C)\sqsubseteq f_{C'}(\mu f_C)$. Joining the two, we obtain
  $f(\mu f_C)\sqsubset \mu f_C$ and thus
  $\mu f_C \notin \mathit{Fix}(f)$, as desired.

  \medskip
  
  (1 $\Leftrightarrow$ 3): By
  Lemma~\ref{le:min-set-basic}(\ref{le:min-set-basic:a}) we know that
  $f(\mu f_C) \sqsubseteq f_C(\mu f_C) = \mu f_C$. Hence
  $\mu f_C \notin \mathit{Fix}(f)$, i.e., $f(\mu f_C ) \neq \mu f_C$
  is equivalent to $f(\mu f_C) \sqsubset \mu f_C$, as desired.
\end{proof}

The above result suggests an algorithm for computing a fixpoint of a
function $f : \monM^Y \to \monM^Y$ on the basis of some
min-decomposition. The idea is to guess some strategy $C$, determine
$\mu f_C$ and check $\minimp{C}$. If this set is empty we have
reached some fixpoint, otherwise choose $C'\in \minimp{C}$ for the
next iteration.
Note that for this algorithm it is irrelevant whether
we use min-improvements or restrict to stable min-improvements.
We also note that this procedure and the developed theory to this point work for monotone functions $f\colon \latL^Y \to \latL^Y$ where $\latL$ is a complete lattice. 

When we are interested in the least fixpoint and the function admits
many fixpoints, the sketched algorithm determines a
fixpoint which might not be the desired one.
Exploiting the theory from~\cite{bekp:fixpoint-theory-upside-down},
summarised in \S\ref{ss:upsidedown}, we can refine the algorithm to
ensure that it computes $\mu f$.
For this, we have to work with non-expansive functions
$f : \monM^Y \to \monM^Y$ with $\monM$ being a complete MV-chain and $Y$ a
finite set.
In fact, in this setting, given a
fixpoint of $f$, say $a \in \monM^Y$, relying on Theorem~\ref{th:fixpoint-sound-compl}, we can check whether it is the
least fixpoint of $f$. In case it is not, we can ``improve''
it obtaining a smaller pre-fixpoint of $f$ in a way that
we can continue the iteration from there.
The resulting algorithm is reported in
Fig.~\ref{fi:alg-min-fix-above}.
Observe that in step~\ref{fi:alg-min-fix-above:2b} we clearly do not
need to compute all improvements. Rather, a min-improvement, whenever it
exists, can be determined, on the basis of
Definition~\ref{de:improvement}, using $\mu f_{C_i}$ computed in
step~\ref{fi:alg-min-fix-above:2a}. Moreover
step~\ref{fi:alg-min-fix-above:2c} relies on
Theorem~\ref{th:fixpoint-sound-compl} and
Lemma~\ref{lem:fp-increase}. 

\begin{figure}
  \begin{tcolorbox}[width=1\textwidth, colframe=white]

    \begin{enumerate}
    \item
      \label{fi:alg-min-fix-above:1}
      Initialize: guess a strategy $C_0$, $i:=0$
    \item \textbf{iterate}
      \label{fi:alg-min-fix-above:2}
      \begin{enumerate}
      \item
        \label{fi:alg-min-fix-above:2a}
        determine $\mu f_{C_i}$ 
      \item
        \label{fi:alg-min-fix-above:2b}
        \textbf{if} $\minimp{C_{i}} \neq \emptyset$, let
        $C_{i+1} \in \minimp{C_i}$; $i:=i+1$; \textbf{goto} (a)
      \item
        \label{fi:alg-min-fix-above:2c}
        \textbf{else if} $\mu f_{C_i} \neq \mu f$ let
        $a\sqsubset \mu f_{C_i}$ be a pre-fixpoint of $f$ and determine
        $C_{i+1}$ via
        \[ C_{i+1}(y) = \arg\min_{h\in H_{\min}(y)} h(a)\]
        $i:=i+1$; \textbf{goto} 2.(a)        
      \item
        \label{fi:alg-min-fix-above:2d}
        \textbf{else} \textbf{stop}
      \end{enumerate}
    \end{enumerate}
  \end{tcolorbox}
  \caption{Computing the least fixpoint, from above}
  \label{fi:alg-min-fix-above}
  
\end{figure}

\begin{theoremrep}[least fixpoint, from above]
  \label{thm:algo-above}
  Let $Y$ be a finite set, $\monM$ a complete MV-chain,
  $f : \monM^Y \to \monM^Y$ be a non-expansive function and let
  $H_{\min}$ be a min-decomposition of $f$. The algorithm in
  Fig.~\ref{fi:alg-min-fix-above} terminates and computes $\mu f$.
\end{theoremrep}

\begin{proof}
  We first argue that in the iteration eventually
  $\minimp{C_i} = \emptyset$, hence we will reach step~\ref{fi:alg-min-fix-above:2c}, and this
  happens iff we have reached some fixpoint of $f$.

  In fact, for any $i$, if
  % $\mu f_{C_i}$ is not a fixpoint of $f$, then, by
  % Theorem~\ref{thm:iff-above},
  $\minimp{C_i} \neq \emptyset$, then one considers
  $C_{i+1} \in \minimp{C_i}$ and by Lemma~\ref{le:improve-above}
  $\mu f_{C_{i+1}}\sqsubset \mu f_{C_i}$. Thus the algorithm computes
  a strictly descending chain $\mu f_{C_i}$. This means that we
  certainly generate a new strategy at each iteration.
  Since there are only finitely many strategies, the iteration must stop
  at some point. When this happens we will have
  $\minimp{C_i} = \emptyset$ and thus, we will reach step~\ref{fi:alg-min-fix-above:2c}.
  By Lemma~\ref{lem:iff-above}, this happens if and only if
  $\mu f_{C_i}$ is a fixpoint of $f$.

  Now, by Theorem~\ref{th:fixpoint-sound-compl}, we
  can determine whether $\mu f_{C_i}$ and thus the algorithm only
  terminates when $\mu f$ is computed.

  Otherwise one considers $a \sqsubset \mu f_{C_i}$, a (pre-)fixpoint of
  $f$, which is given by
  Lemma~\ref{lem:fp-increase}.

  Due to Lemma~\ref{le:min-set-basic}(\ref{le:min-set-basic:b}) and
  the way the new strategy $C_{i+1}$ is defined we know that
  $f_{C_{i+1}}(a)(y) = C_{i+1}(y)(a) = f(a)(y) \sqsubseteq a(y)$ for
  all $y \in Y$, we have that $a$ is a (pre-)fixpoint also of $f_{C_{i+1}}$, and so
  $\mu f_{C_{i+1}} \sqsubseteq a \sqsubset \mu f_{C_i}$.

  Therefore, again we obtain a descending chain of least
  fixpoints. The number of strategies is finite, since $Y$ is finite
  and $H_{\min}(y)$ is finite for all $y \in Y$.  Thus we will at
  some point compute $\mu f_{C_j} = \mu f$ for some strategy $C_j$ and
  terminate since by Theorem~\ref{th:fixpoint-sound-compl} we can
  determine when this is the case.
\end{proof}

Termination easily follows from the fact that the number of
strategies is finite (since $Y$ is finite and $H_{\min}(y)$ is finite for all
$y \in Y$). Given that at any iteration the fixpoint decreases, no
strategy can be considered twice, and thus the number of iterations is
bounded by the number of strategies.

\begin{example}
  Let us revisit Example~\ref{ex:sg} and the fixpoint function
  $\mathcal{V}$ defined there. Its least fixpoint satisfies
  $\mu \mathcal{V}(\state{1}) = 1$ and
  $\mu \mathcal{V}(v) = \varepsilon$ for any
  $v\in V\setminus \{\state{1}\}$.

  The optimal strategy for {\Min} is to choose $\state{av}$
  as its successor since this forces {\Max} to exit the cycle formed
  by $\state{min},\state{av},\state{max}$ to $\state{\varepsilon}$, yielding a payoff of
  $\varepsilon$ for these states. If {\Max} would behave in a way
  that the play keeps cycling he would obtain a payoff of $0$, which
  is suboptimal.

  We now apply our algorithm. We start by guessing a strategy for
  {\Min}, so we assume $C_0(\state{min}) = h_{\state{1}}$, i.e.\
  $C_0 = C_1^{\min}$ (for the naming of the strategies we refer to
  Example~\ref{ex:sg}). The least fixpoint $\mu \mathcal{V}_{C_0}$ can
  be found by solving the following linear program:
  \begin{align*}
    {\min} & \sum_{v\in V} a(v) 
    &
      a(\state{1}) &= 1
    &
      a({\state{\varepsilon}}) &= \varepsilon
    &
      a(\state{av}) &= \frac{1}{2} a(\state{min}) + \frac{1}{2} a(\state{max})\\[-2mm]
    & &
    a(\state{max}) &\ge a({\state{\varepsilon}})
    &
      a(\state{max}) &\ge a(\state{av})
    &
      a(\state{min}) &= a(\state{1})
  \end{align*}
  with $0\le a(v) \le 1$ for $v\in V$, which yields
  $\mu \mathcal{V}_{C_0} ({\state{\varepsilon}}) = \varepsilon$ and
  $\mu \mathcal{V}_{C_0}(v) = 1$ for all
  $v\in V\setminus \{{\state{\varepsilon}}\}$. Now
  $\mu \mathcal{V}_{C_0}$ is a fixpoint of $\mathcal{V}$ -- but not
  the least -- and thus we find the vicious cycle formed by
  $\state{min}, \state{av},\state{max}$, i.e.\
  $\nu \mathcal{V}_\#^{\mu \mathcal{V}_{C_0}} = \{
  \state{min},\state{av},\state{max}\}$ and decrease the values of
  those states in $a$ by $\delta$, i.e.\ we obtain
  $a = \mu \mathcal{V}_{C_0} \ominus \delta_{\{\state{min},
    \state{av},\state{max}\}}$. This results in $a(\state{1}) = 1$,
  $a({\state{\varepsilon}}) = \varepsilon$ and $a(v) = 1-\delta$ for
  all $v\in V\setminus \{ \state{1}, \state{\varepsilon}\}$. Any
  $\delta\in (0,1-\varepsilon]$ is a valid choice.

  Computing $C_{1}(y) = \arg\min_{h\in H_{\min}(y)} h(a)$ yields the strategy
  $C_1 = C_2^{\min}$, i.e.\ $C_1(\state{min}) = h_{\state{av}}$. By
  linear programming (replace $a(\state{min}) = a(\state{1})$ by
  $a(\state{min}) = a(\state{av})$) we obtain
  $\nu \mathcal{V}_\#^{\mu f_{C_1}} = \emptyset$, thus
  $\mu \mathcal{V}_{C_1} = \mu \mathcal{V}$ and the algorithm
  terminates.
\end{example}

\subsection{Strategy iteration from below}
\label{SIfromAbove}

Here we present a different generalized strategy iteration algorithm approaching
the least fixpoint from below. Intuitively, now it is player {\Max}
who improves his strategy step by step, creating an ascending chain
of least fixpoints which reaches the least fixpoint
of the underlying function $f$. Despite the fact that in this case we
cannot get stuck at a fixpoint which is not the least, the correctness
argument is more involved.

We will deal with max-decompositions of a function and we will
need a notion of (stable) max-improvement which is naturally defined
as a dualisation
of the notion of (stable) min-improvement (Definition~\ref{de:improvement}).

\begin{definition}[max-improvement]
  \label{de:max-improvement}
  Let $Y$ be a finite set, $\monM$ a complete MV-chain,
  $f : \monM^Y \to \monM^Y$ a monotone function and $H_{\max}$ a
  max-decomposition. Given $C, C'$ strategies in $H_{\max}$, we say
  that $C'$ is a \emph{max-improvement} of $C$ if
  $\mu f_C \sqsubset f_{C'}(\mu f_C)$. It is called a \emph{stable}
  max-improvement if in addition $C'(y) = C(y)$ for all $y \in Y$ such
  that $f_{C'}(\mu f_C)(y) = \mu f_C(y)$. We denote by $\maximp{C}$
  (respectively $\maximps{C}$) the set of (stable) max-improvements of
  $C$.
\end{definition}

When iterating from above it was rather easy to show that given a
strategy $C$ and a min-improvement $C'$, the latter yields a smaller
least fixpoint $\mu f_{C'}\sqsubset \mu f_C$ (Lemma~\ref{le:improve-above}).
Observing that $\mu f_C$
is a pre-fixpoint of $f_{C'}$ was enough to prove this.

Here, however, we cannot simply dualise the argument. If $C'$ is a
max-improvement of $C$, we obtain that $\mu f_C$ is a post-fixpoint of
$f_{C'}$ which, in general, does not guarantee $\mu f_{C'}\sqsupset
\mu f_C$.
We have to resort to stable max-improvements and, in order to show
that such improvements in fact yield greater least fixpoints, we
need, again, to use the theory reviewed in
\S\ref{ss:upsidedown}. Hence, we have to work with
non-expansive functions $f\colon \mathbb{M}^Y\to \mathbb{M}^Y$ where
$\mathbb{M}$ is a complete MV-chain.

\begin{toappendix}
In order to prove that a stable max-improvement leads to greater least fixpoint,
we first need a generalisation of the soundness result in
Theorem~\ref{th:fixpoint-sound-compl} to the case
in which $a$ is a post-fixpoint instead of a fixpoint,
proved in~\cite{bekp:fixpoint-theory-upside-down} and recalled below.

\begin{theorem}[soundness for post-fixpoints
  \cite{bekp:fixpoint-theory-upside-down}]
  \label{thm:soundness}
  Let $\monM$ be a complete
  MV-chain, $Y$ a finite set and $f : \monM^Y\to \monM^Y$ be a
  non-expansive function.
  Given a post-fixpoint $a \in \monM^Y$ of $f$, let
  $\Ytop{Y}{a=f(a)} = \{ y \in \Ytop{Y}{a} \mid a(y) = f(a)(y) \}$.
  Let us 
  define $f^a_* : \Ytop{Y}{a=f(a)} \to \Ytop{Y}{a=f(a)}$ as
  $f_*^a(Y') = f^a_{\#}(Y') \cap \Ytop{Y}{a=f(a)}$, where
  $f^a_{\#}: \Pow{\Ytop{Y}{a}} \to \Pow{\Ytop{Y}{f(a)}}$ is the $a$-approximation of $f$.
  If $\nu f^a_* = \emptyset$
  then $a \sqsubseteq \mu f $.
\end{theorem}
\end{toappendix}

\begin{lemmarep}[max-improvements increase fixpoints] 
  \label{lem:improve-below}
  Let $Y$ be a finite set, $\monM$ a complete MV-chain,
  $f : \monM^Y \to \monM^Y$ a non-expansive function and $H_{\max}$ a
  max-decomposition. Given a strategy $C$ in $H_{\max}$ and a stable
  max-improvement $C' \in \maximps{C}$, then
  $\mu f_{C} \sqsubset \mu f_{C'}$.  
\end{lemmarep}

\begin{proof}
  We show that
  $\nu(f_{C'})_*^{\mu f_C} = \emptyset$ and thus conclude that
  $\mu f_{C} \sqsubseteq \mu f_{C'}$ by
  Theorem~\ref{thm:soundness}. Recall that
  $(f_{C'})_*^{\mu f_C}\colon
  \Ytop{Y}{\mu f_C = f_{C'}(\mu f_C)} \to
  \Ytop{Y}{\mu f_C = f_{C'}(\mu f_C)}$, i.e., it restricts to those
  elements of $\mu f_C$ where $\mu f_C$ and $f_{C'}(\mu f_C)$ coincide.
    
  We proceed by showing that $(f_C)_*^{\mu f_C}$ and
  $(f_{C'})_*^{\mu f_C}$ agree on $\Ytop{Y}{\mu f_C=f_{C'}(\mu f_C)}$,
  which, by definition, is a subset of
  $\Ytop{Y}{\mu f_C} = \Ytop{Y}{\mu f_C=f_C(\mu f_C)}$ (remember that
  $\mu f_C$ is a fixpoint of $f_C$). It holds that
  \begin{eqnarray*}
    (f_C)_*^{\mu f_C}(Y') & = &
        \gamma^{f_C(\mu f_C),\delta}(f_C(
        \alpha^{\mu f_C,\delta}(Y'))) \\
      (f_{C'})_*^{\mu f_C}(Y') & = &
        \gamma^{f_{C'}(\mu f_C),\delta}(f_{C'}(
        \alpha^{\mu f_C,\delta}(Y'))) \cap
        \Ytop{Y}{\mu f_C=f_{C'}(\mu f_C)}
  \end{eqnarray*}
  for a suitable constant $\delta$ and if we choose $\delta$ small
  enough we can use the same constant in both cases.
    
  Now let
  $y\in \Ytop{Y}{\mu f_C=f_{C'}(\mu f_C)}$: by definition it
  holds that
  \begin{eqnarray*}
    && y\in (f_C)_*^{\mu f_C}(Y') = \gamma^{f_C(\mu f_C),\delta}(f_C(
    \alpha^{\mu f_C,\delta}(Y'))) \\
    & \iff &
    f_C(\mu f_C)(y) \ominus f_C(\alpha^{\mu f_C,\delta}(Y'))(y)
    \sqsupseteq \delta \\
    & \iff & C(y)(\mu f_C) \ominus C(y)(\alpha^{\mu f_C,\delta}(Y'))
    \sqsupseteq \delta
  \end{eqnarray*}
  
  Now, since $\mu f_C(y) = f_{C'}(\mu f_C)(y)$, by definition of
  stable max-improvement (remember that we require
  $C'\in \maximps{C}$) we have $C(y) = C'(y)$, and thus
  \begin{eqnarray*}
    && C(y)(\mu f_C) \ominus C(y)(\alpha^{\mu f_C,\delta}(Y'))
    \sqsupseteq \delta \\
    & \iff & C'(y)(\mu f_C) \ominus C'(y)(\alpha^{\mu f_C,\delta}(Y'))
    \sqsupseteq \delta \\
    & \iff & y\in (f_{C'})_*^{\mu f_C}(Y')
  \end{eqnarray*}
  Thus
  $\nu (f_{C'})_*^{\mu f_C}\subseteq \nu (f_{C})_*^{\mu
    f_C}=\emptyset$.
    
  In conclusion, we have that $\mu f_{C'} \sqsupseteq \mu f_C$
  and since $\mu f_C$ is not a fixpoint of $f_{C'}$ (because $C'$ is a
  max-improvement of $C$) we conclude $\mu f_{C'} \sqsupset \mu f_C$.
\end{proof}

\begin{example}
  We note that working with max-improvements which are stable is
  essential for the validity of Lemma~\ref{lem:improve-below} above.
  In fact, consider the SSG in Figure~\ref{fi:ssg-example-imp} where
  $\state{max}_1,\state{max}_2\in \mathit{MAX}$ and
  $\state{1} \in \mathit{SINK}$, with reward $1$. Let $C$ be the
  strategy for {\Max} where $\state{max}_1$ and $\state{max}_2$ have
  as successors $\state{1}$ and $\state{max}_2$, respectively. It is
  easy to see that
  $\mu \mathcal{V}_C(\state{1}) = \mu \mathcal{V}_C(\state{max}_1) =
  1$ and $\mu \mathcal{V}_C (\state{max}_2) = 0$. Now, an improvement
  in $\mathit{imp}_{\max}(C)$ can be the strategy $C'$ which chooses
  $\state{max}_1$ as a successor for both $\state{max}_1$ and
  $\state{max}_2$. Then we have
  $\mu \mathcal{V}_{C'}(\state{max}_1)= \mu
  \mathcal{V}_{C'}(\state{max}_2)=0$, hence
  $\mu \mathcal{V}_C \sqsupset \mu \mathcal{V}_{C'}$. The reason why
  this happens is that $C'$ is not a stable improvement of $C$ since
  it uselessly changes the successor of $\state{max}_1$ from
  $\state{1}$ to $\state{max}_1$, both mapped to $1$ by
  $\mu \mathcal{V}_C$.
  A stable improvement of $C$ is $C''$ where $\state{max}_1$ and
  $\state{max}_2$ have as successors $\state{1}$ and $\state{max}_1$,
  respectively. Then it can be seen that  $\mu \mathcal{V}_{C''}(v)= 1$ for all states.
    
  \begin{figure}[t]
    \begin{center}
      \begin{tikzpicture}
        \node (S1) at (0,0) [circle,draw,minimum width=2.5em,inner sep=0] {$\state{1}$};
        \node (MAX1) at (2,0) [circle,draw,minimum width=3.5em,inner sep=0] {$\state{max}_1$};
        \node (MAX2) at (4.2,0) [circle,draw,minimum width=3.5em,inner sep=0] {$\state{max}_2$};

        \draw[->] (MAX1) to (S1);
        \draw[->,bend left] (MAX2) to (MAX1);
	    	\draw[->,out=20,in=50,looseness=8] (MAX1) to (MAX1);
       	\draw[->,out=15,in=345,looseness=8] (MAX2) to (MAX2);
      \end{tikzpicture}
    \end{center}
    \caption{An example of a simple stochastic game where state
      $\state{1}$ has payoff $1$.}
    \label{fi:ssg-example-imp}
  \end{figure}
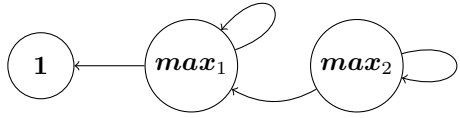
\end{example}

Relying on Lemma~\ref{lem:improve-below}, we can easily prove the
dual of Lemma~\ref{lem:iff-above}, showing that a strategy admits a
stable max-improvement as long as we have not reached a fixpoint of
$f$.

\begin{lemmarep}[max-improvements exist for non-fixpoints]
\label{lem:iff-below}
Let $Y$ be a finite set, $\monM$ a complete MV-chain,
$f : \monM^Y \to \monM^Y$ a monotone function and $H_{\max}$ a
max-decomposition. Given a strategy $C$ in $H_{\max}$, the following
are equivalent:
  \begin{enumerate}
  \item $\mu f_C \not\in \mathit{Fix}(f)$
  \item $\maximps{C} \neq \emptyset$
  \item $\mu f_C \sqsubset f(\mu f_C)$
  \end{enumerate}
\end{lemmarep}

\begin{proof}
  Adapt Lemma~\ref{lem:iff-above}, observing that
  whenever there is a max-improvement, i.e.\linebreak $\maximp{C} \neq \emptyset$,
  then there exists one which is stable, i.e.\ $\maximps{C} \neq \emptyset$,
  and, clearly, vice versa.
  This holds since, given a strategy $C' \in \maximp{C}$, for all $y \in Y$,
  we can define $C''(y) = C'(y)$ when
  $\mu f_C(y) \sqsubset f_{C'}(\mu f_C)(y)$, and $C''(y) = C(y)$
  otherwise. Then $\mu f_{C} \sqsubset f_{C''}(\mu f_C)$ and
  $C'' \in \maximps{C}$.
\end{proof}

To summarise, given a strategy $C$ with
$\mu f_C \notin \mathit{Fix}(f)$ we can construct a strategy $C'$ with
$\mu f_C\sqsubset \mu f_{C'}$. This creates an ascending
chain of least fixpoints and since there are only finitely many
strategies we will at some point find an optimal strategy $C^*$ with
$\mu f_{C^*}=\mu f$.

\begin{propositionrep}[least fixpoint from max-decomposition]
  \label{pr:below-muf}
  Let $Y$ be a finite set, $\monM$ a complete MV-chain,
  $f : \monM^Y \to \monM^Y$ a non-expansive function and let
  $H_{\max} : Y \to \Powf{\monM^Y \to \monM}$ be a max-decomposition of $f$.  
  Then $\mu f = \max \{\mu f_C \mid C \text{ is a strategy in } H_{\max}\}$.
\end{propositionrep}

\begin{proof}
  By Lemma~\ref{le:min-set-basic}(\ref{le:min-set-basic:a}) (adapted
  to the max case) we have $f_C\sqsubseteq f$ for any strategy $C$ in
  $H_{\max}$. Therefore, $\mu f_C \sqsubseteq \mu f$ and thus
  \begin{center}
    $\mu f \sqsupseteq \max \{\mu f_C \mid C \text{ is a strategy in } H_{\max}\}$.
  \end{center}
  Assume, by contradiction, that the inequality above is strict. Then
  for all strategies $C$ in $H_{\max}$ we have
  $\mu f_C\sqsubset \mu f$.  Then, by Lemma~\ref{lem:iff-below},
  each strategy admits a stable max-improvement. Starting from any
  strategy $C_0$ one could thus generates a sequence of stable
  max-improvements $C_1$, $C_2$, \ldots. Since by
  Lemma~\ref{lem:improve-below},
  $\mu f_{C_i} \sqsubset \mu f_{C_{i+1}}$, all these improvements
  would be different, thus contradicting the finiteness of the
  max-decomposition and hence the fact that there are finitely many
  strategies.
\end{proof}

The above results lead us to a generalised strategy iteration
algorithm which approaches the least fixpoint from below.

\begin{figure}[t]
  \begin{tcolorbox}[width=1\textwidth, colframe=white]
    \begin{enumerate}
    \item Initialize: guess a strategy $C_0$, $i:=0$
    \item \textbf{iterate}
      \begin{enumerate}
      \item determine $\mu f_{C_i}$
      \item \textbf{if} $\maximps{C_{i}} \neq \emptyset$, let
        $C_{i+1} \in \maximps{C_i}$; $i:=i+1$; \textbf{goto} (a)
      \item else \textbf{stop}
      \end{enumerate}
    \end{enumerate}
  \end{tcolorbox}
  \caption{Computing the least fixpoint, from below}
  \label{fi:alg-min-fix-below}
\end{figure}

% This algorithm indeed computes $\mu f$ since it creates an ascending
% chain of least fixpoints. Since there are only finitely many
% strategies an optimal strategy $C_j$ will be found at some point,
% i.e.\ $\mu f_{C_j} = \mu f$.

\begin{theoremrep}[least fixpoint, from below]
  \label{thm:algo-below}
  Let $Y$ be a finite set, $\monM$ a complete MV-chain,
  $f : \monM^Y \to \monM^Y$ be a non-expansive function and let
  $H_{\max}$ be a max-decomposition of $f$. The algorithm in
  Fig.~\ref{fi:alg-min-fix-below} terminates and computes $\mu f$.
\end{theoremrep}

\begin{proof}
  For any $i$, if $\mu f_{C_i}$ is not a fixpoint of $f$, then, by
  Lemma~\ref{lem:iff-below}, $\maximps{C_i} \neq \emptyset$. By
  Lemma~\ref{lem:improve-below} $\mu f_{C_{i}}\sqsubset \mu f_{C_{i+1}}$
  for any $C_{i+1} \in \maximps{C_i}$. Thus the algorithm computes
  a strictly ascending chain $\mu f_{C_i}$. This means that we
  certainly generate a new strategy at each iteration.

  Since there are only finitely many strategies
  (as argued in the proof of Theorem~\ref{thm:algo-above}), the iteration must stop at some
  point. When this happens we will have $\maximps{C_i} = \emptyset$
  and thus, by Lemma~\ref{lem:iff-below}, $\mu f_{C_i}$ is a
  fixpoint of $f$. By Proposition~\ref{pr:below-muf} we conclude that
  $\mu f_{C_i} = \mu f$.
\end{proof}

The iteration from below may seem more appealing since it cannot get
stuck at any fixpoint of $f$. However, it has to be noted that the
computation of $\mu f_C$ -- for a chosen strategy $C$ - may be more
difficult than before, which is illustrated by the following example.

\begin{example}
  Let us apply the above algorithm to the SSG in
  Example~\ref{ex:sg}. Recall that the least fixpoint is given by
  $\mu \mathcal{V}(\state{1}) = 1$ and $\mu \mathcal{V}(v) = \varepsilon$ for all
  $v\in V\setminus \{ \state{1}\}$.
  
  We start by guessing a strategy for {\Max}, so we assume
  $C_0(\state{max}) = h_{\state{av}}$, i.e.\ $C_0 = C_2^{\state{max}}$. With
  this choice of strategy, {\Min} is able to keep the game going
  infinitely in the cycle formed by $\state{min},\state{av},\state{max}$ and thus
  payoff~$0$ is obtained. Now $\mu \mathcal{V}_{C_0}$ is given by
  $\mu \mathcal{V}_{C_0} (\state{\varepsilon})=\varepsilon$,
  $\mu \mathcal{V}_{C_0}(\state{1})$ $=1$ and
  $\mu \mathcal{V}_{C_0}(v) = 0$ for all
  $v\in V\setminus \{ {\state{\varepsilon}},\state{1}\}$. We note that
  $\mu \mathcal{V}_{C_0}$ cannot immediately be computed via linear
  programming, but there is a way to modify the fixpoint equation to
  have a unique fixpoint and hence linear programming can be used
  again \cite{bekp:fixpoint-theory-upside-down}. This is done by
  precomputing states from which {\Min} can force a non-terminating
  play and assigning payoff value $0$ to them.
  Next, {\Max} updates his strategy and we obtain $C_1 =
  C_1^{\max}$. As above we can compute
  $\mu \mathcal{V}_{C_1}$ -- which, this time, equals $\mu
  \mathcal{V}$ -- via linear programming.
\end{example}

\begin{remark}
  Given $\mu f$ (without the corresponding strategy) an interesting
  question is how one can derive optimal strategies for {\Min} or
  {\Max}. Note that each presented strategy iteration algorithm only
  produces an optimal strategy for one player, but not for the other.

  It is rather easy to find an optimal strategy with respect to
  $H_{\min}$. We can simply compute
  $C^*(y) = \arg\min_{h\in H_{\min}(y)} h(\mu f)$ which yields
  some optimal strategy $C^*$, i.e.\ $\mu f_{C^*} = \mu f$. It is
  enough to choose some minimum, even if this is ambiguous and there
  are several choices, each of which produces an optimal
  strategy. The strategy $C^*$ is optimal since $\mu f$ is a pre-fixpoint of
  $f_{C^*}$ and $\mu f = \mu f_{C^*}$ follows from
  Proposition~\ref{pr:above-muf}.

  On the other hand, given $\mu f$, we cannot easily obtain an optimal
  strategy in $H_{\max}$. We will discuss in \S\ref{ex:EG}
  (Example~\ref{ex:EG}) that defining
  $C^*(y) = \arg\max_{h\in H_{\max}(y)} h(\mu f)$ for an
  arbitrary $h$ where the value is maximal does not work in
  general.
\end{remark}

\section{Application: energy games}
\label{se:EG}

In this section we examine energy games
\cite{DBLP:conf/icalp/DorfmanKZ19} and show how both strategy
iterations in \S\ref{se:GSI} can be applied to solve energy games and
have the advantage of providing us not only with the value vector, but
also with a strategy, which is interesting, in particular, for
Player~1 ({\Max}).

Energy games are two-player games on finite graphs where Player~0
({\Min}) wants to keep the game going forever, while Player~1 ({\Max})
wants it to stop eventually. Each state belongs to one player where
he chooses the successor and each traversed edge reduces or increases
the energy level by some integer. The game stops when an edge is taken
which reduces the energy level below 0. The main question
is how
much initial energy is needed, such that Player~0 can keep the game
going forever.
It is possible to require an initial energy level of infinity.

\subsection{Introduction to energy games}
\label{se:eg-intro}

A game graph is a tuple $\Gamma = (V_0,V_1,E,w)$ where
$V = V_0\cup V_1$, with $V_0 \cap V_1 = \emptyset$, is the set of
states, $E\subseteq V\times V$ is the set of edges and
$w\colon E \to \mathbb{Z}$ is a weight function. We assume that each
state has at least one outgoing edge. We define
$\mathit{post}(v) = \{ v'\in V \mid (v,v')\in E\} \neq \emptyset$ for $v\in V$.
States in $V_0$ and $V_1$ are owned by Player~0 and Player~1
respectively. Moving from state $v$ to state $v'$ will change the
energy level by adding the value $w(v,v')$. If this value is positive,
some energy is gained, otherwise energy decreases. It is the aim of
Player~0 to keep the energy level from getting negative. An energy
game is an infinite play on a game graph $\Gamma$ and we note that
optimal positional strategies exist for both players~\cite{DBLP:conf/icalp/DorfmanKZ19}.

The solution of $\Gamma$ is a function
$g_\Gamma\colon V\to \mathbb{N}^\infty$ (where
$\mathbb{N}^\infty = \mathbb{N}\cup \{ \infty \}$) that assigns to
each state the least energy level which is sufficient for Player~0 to
keep the game going, independently of the chosen strategy of
Player~1. It is known that the solution is the least fixpoint of the
following function
$\bar{\mathcal{E}} \colon (\mathbb{N}^\infty)^V \to
(\mathbb{N}^\infty)^V$, defined as
\begin{align*}
  \bar{\mathcal{E}}(a)(v) =
  \begin{cases}
    \min\limits_{v' \in \mathit{post}(v)} \max\{a(v') - w(v,v'),0\} & \mbox{if }v\in V_0 \\
    \max\limits_{{v' \in \mathit{post}(v)}} \max\{a(v') - w(v,v'),0\} & \mbox{if }v\in V_1
  \end{cases}
\end{align*}
         
\begin{example}
  \label{ex:EG}
  Consider the following energy game, where it is intended that
  circular and rectangular states belong to Player~0 and Player~1,
  respectively.
  
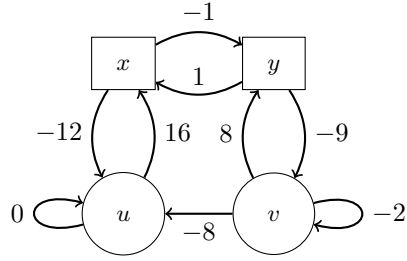
\begin{center}
      \begin{tikzpicture}
        \node (x) at (0,0) [rectangle,draw] {\begin{tabular}{c}
            $x$
          \end{tabular} };
        \node (u) at (0,-2) [circle,draw] {\begin{tabular}{c}
            $u$
          \end{tabular} };
        \node (y) at (2,0) [rectangle,draw] {\begin{tabular}{c}
            $y$
          \end{tabular} };
        \node (v) at (2,-2) [circle,draw] {\begin{tabular}{c}
            $v$
          \end{tabular} };
        
         \draw[->,thick,bend right] (x) to node [left]  {$-12$} (u);
        \draw[->,thick,bend right] (u) to node [right]  {$16$} (x);
        \draw[->,thick] (v) to node [below]  {$-8$} (u); 
        \draw[->,thick,   loop left] (u) to node [left]  {$0$} (u);
        \draw[->,thick,   loop right] (v) to node [right]  {$-2$} (v);
        \draw[->,thick,bend left] (x) to node [above]  {$-1$} (y);
        \draw[->,thick,bend left] (y) to node [above]  {$1$} (x);
        \draw[->,thick,bend left] (y) to node [right]  {$-9$} (v);
        \draw[->,thick,bend left] (v) to node [left]  {$8$} (y);        
      \end{tikzpicture}
  \end{center}
  
  The optimal strategy for Player~0 is
  to choose $u$ as the successor to $u$ and $v$. Thus $v$
  requires an initial energy of 8 to keep going forever. For $u$ an
  initial energy of 0 is sufficient.
  On the other hand, the optimal strategy for Player~1 is to choose
  $y$ as successor to $x$ and $v$ as successor to $y$. This results in
  a required initial energy of $17$ for $y$ and $18$ for $x$.

  Thus, we obtain as least fixpoint
  $g_\Gamma(x)=18,~ g_\Gamma(y) = 17,~ g_\Gamma(u) = 0,~ g_\Gamma(v) =
  8$. Note that, if from $u$ Player~0 would choose $x$, Player~1
  could keep the game in a negative cycle.

  Given only the least fixpoint $g_\Gamma$, the strategy
  of Player~1 is not deducible with a local reasoning,
  since from $y$ the choices $x,v$ are indistinguishable
  (in fact $g_\Gamma(y) = 17 = g_\Gamma(x)-1 = g_\Gamma(v)-(-9)$).
  However, if $x$ is chosen as
  successor to $y$ (and still $y$ as successor to $x$), we end up in a
  value vector where {\Min} needs $0$ initial energy in $y$ to
  keep going.
\end{example}

% Working on $\mathbb{N}^\infty$ has the drawback that this set is
% infinite.
% % This can be remedied
% Observing that for each $v\in V$ either
% $g_\Gamma(v) < n\cdot W$ (for $n=|V|$ and $W=\max_{e\in E} |w(e)|$) or
% $g_\Gamma(v) = \infty$~\cite{DBLP:conf/icalp/DorfmanKZ19}, we can actually
% define the fixpoint equation
% on a finite set $\{0,\dots,n\cdot W,\top\}$, where $\top$ is
% returned whenever the value exceeds $n\cdot W$,
% % which denotes the fact
% indicating
% that Player~0 does not have enough energy at her  disposal.

\subsection{Strategy iteration for energy games}
\label{se:strategy-iteration-eg}

In order to solve energy games in our framework, we have to consider
non-expansive functions over MV-algebras, however
  $\mathbb{N}^\infty$ is unfortunately not an MV-algebra. For this,
we use the results of~\cite{DBLP:conf/icalp/DorfmanKZ19}, where it is
shown that any energy game $\Gamma=(V_0,V_1,E,w)$ can be transformed
into an energy game $\Gamma'=(V_0',V_1',E',w')$ with finite values
only.
Concretely, this is done by adding an
``emergency exit'' for each state in $V_0$ guaranteeing a finite
amount of required energy to keep the game going forever.  The
solution $g_{\Gamma'}$ of $\Gamma'$ satisfies
$g_{\Gamma'}(v) < \infty$ for all $v\in V$ and the solution $g_\Gamma$
of $\Gamma$ can be easily reconstructed from $g_{\Gamma'}$. This
allows us to restrict to energy games with finite values,
where the solution is bounded by a suitable $k$.
In this setting,
letting $K=\{0,\dots, k\}$ and
$\ominus_\mathbb{Z}\colon K \times \mathbb{Z} \to K$ given
by $x\ominus_\mathbb{Z} y = \min \{ \max \{ x - y,0 \}, k\}$,
we can define $\mathcal{E} \colon K^V \to K^V$ for
$a\colon V\to K$ and $v\in V$ as 
\begin{align*}
  \mathcal{E} (a)(v) =
  \small
  \begin{cases}
    \min\limits_{(v,v')\in E} a(v')\ominus_\mathbb{Z} w(v,v') &
    \mbox{if }v\in V_0 \\
    \max\limits_{(v,v')\in E} a(v')\ominus_\mathbb{Z} w(v,v') &
    \mbox{if }v\in V_1
  \end{cases}
\end{align*}

\begin{lemmarep}[solution is least fixpoint of $\mathcal{E}$]
  \label{lem:E=g}
  Let $\Gamma$ be an energy game with finite values,
  %where every value
  bounded by $k$. Then $\mu \mathcal{E}=g_\Gamma$, i.e.\ the least
  fixpoint of $\mathcal{E}$ coincides with the solution of~$\Gamma$.
\end{lemmarep}

\begin{proof}
  First note that from known results on energy games
  \cite{DBLP:journals/fmsd/BrimCDGR11} the solution $g_\Gamma$ equals
  $\mu \bar{\mathcal{E}}$, the least fixpoint of $\bar{\mathcal{E}}$.
  
  The claim follows straightforwardly from the fact that, by
  distributivity, $\mathcal{E}(a) = \min\{
  \bar{\mathcal{E}}(a),k\}$. Hence $\mathcal{E} \le \bar{\mathcal{E}}$
  and so $\mu \mathcal{E} \le \mu \bar{\mathcal{E}}$. For the other
  direction observe that by assumption the solution $g_\Gamma$ is
  bounded by $k$ ($g_\Gamma \le k$). Hence
  $\bar{\mathcal{E}}(\mu \mathcal{E}) \le \bar{\mathcal{E}}(\mu
  \bar{\mathcal{E}}) = \mu \bar{\mathcal{E}} = g_\Gamma \le k$. Hence
  $\bar{\mathcal{E}}(\mu \mathcal{E}) = \min\{\bar{\mathcal{E}}(\mu
  \mathcal{E}),k\} = \mathcal{E}(\mu \mathcal{E}) = \mu \mathcal{E}$,
  which means that $\mu \mathcal{E}$ is some fixpoint of
  $\bar{\mathcal{E}}$, implying that
  $\mu \bar{\mathcal{E}} \le \mu \mathcal{E}$. \qedhere
\end{proof}

Recall from Example~\ref{ex:mv-chains} that $K$ is an
MV-chain.
\begin{toappendix}
  
\subsection{Deriving the approximation for energy games}
\label{se:approximation-eg}

In this section we show that $\mathcal{E} : K^V \to K^V$ is
non-expansive and we derive the approximation $\mathcal{E}^a_\#$ which is
required when iterating from above.  First, we disassemble
$\mathcal{E}\colon K^{V}\to K^{V}$ into smaller functions as explained
in Appendix~\ref{se:decomposition-approximation-f}.

We define
\begin{itemize}
\item $E_0 = \{ (v,v')\in E \mid v\in V_0 \}$ and $E_1 = \{
  (v,v')\in E \mid v\in V_1 \}$.
  
  Immediately we conclude $E_0 \cup E_1 = E$ and
  $E_0 \cap E_1 = \emptyset$.
\item \emph{Projections:} $\pi _i \colon E\to V$, $i=1,2$, where,
  given $e=(v,v')\in E$ we have $\pi_1(e) = v$ and $\pi_2(e) = v'$.
\item \emph{Restricted projections:} We define $\pi _i^j \colon E_j\to
  V$ where $\pi_i^j = (\pi_i)_{|_{E_j}}$ for $i\in\{1,2\}$, $j\in\{0,1\}$.
  % $i=1,2$, $j=0,1$, where, given $e=(v,v')\in E_0$ we have
  % $\pi_1^0(e) = v\in V_0$ and $\pi_2^0(e) = v'\in V$ and given
  % $e=(v,v')\in E_1$ we have $\pi_1^1(e) = v\in V_1$ and
  % $\pi_2^1(e) = v'\in V$.
\item \emph{Subtraction of edge weights:} Given $w\colon E\to Z$,
  define $\mathit{sub}_w\colon K^E\to K^E$ via
  \[ \mathit{sub}_w (a)(e) = a(e) \ominus_\mathbb{Z} w(e) \] for
  $a\colon E\to K$ and $e\in E$.
\item \emph{Minimum and maximum functions:} We use the functions
  $\min_u$, $\max_u$ from
  Table~\ref{tab:basic-functions-approximations}, where $u$ is one of
  the projections defined above.
%   Define $\min\nolimits_{\pi_1^0} \colon K^{E_0} \to K^{V_0}$ via
%   \[ \min\nolimits_{\pi _1^0}(f)(v) = \min\limits_{\pi_1^0(e) = v} f(e) \] for
%   $f\colon E_0 \to K$ and $v\in V_0$.
% \item \emph{Max function:} Define
%   $\max\nolimits_{\pi_1^1} \colon K^{E_1} \to K^{V_1}$ via
%   \[ \max\nolimits_{\pi _1^1}(f)(v) = \max\limits_{\pi_1^1(e) = v} f(e) \]
%   for $f\colon E_1 \to K$ and $v\in V_1$.
\end{itemize}

\begin{lemma}
  The function $\mathcal{E} \colon K^V\to K^V$ can be written as
  \[ \mathcal{E} = (\min\nolimits_{\pi_1^0} \uplus \max\nolimits_{\pi_1^1}) \circ
    \mathit{sub}_w \circ \pi_2^*\]
\end{lemma}

\begin{proof}
  Given $a \colon V \to K$ and $v\in V$, we get
  \begin{align*}
    \big((\min\nolimits_{\pi_1^0} \uplus \max\nolimits_{\pi_1^1})
    \circ \mathit{sub}_w \circ \pi_2^*\big)(a)(v)
%    &= (\min_{\pi_1^0} \uplus \max_{\pi_1^1}) \circ \mathit{sub}_w (a)(\pi_2(v,v')) \\
    &= \begin{cases} \min\limits_{\pi_1^0 (e) = v} (\mathit{sub}_w \circ \pi_2^*)(a)(e) \\
      \max\limits_{\pi_1^1 (e) = v} (\mathit{sub}_w \circ \pi_2^*)(a)(e) \end{cases} \\
%          &= \begin{cases} \min\limits_{(v,v')\in E_0} \mathit{sub}_w \circ \pi_2^*(a)(e) \\
%      \max\limits_{(v,v')\in E_1} \mathit{sub}_w \circ \pi_2^*(a)(e),  \end{cases} \\
                &= \begin{cases} \min\limits_{(v,v')\in E} \mathit{sub}_w \circ \pi_2^*(a)(v,v'), &\text{if}~v\in V_0 \\
      \max\limits_{(v,v')\in E} \mathit{sub}_w \circ \pi_2^*(a)(v,v'), &\text{if}~v\in V_1 \end{cases} \\
                      &= \begin{cases} \min\limits_{(v,v')\in E}  \pi_2^*(a)(v,v')\ominus_\mathbb{Z} w(v,v'), &\text{if}~v\in V_0 \\
      \max\limits_{(v,v')\in E} \pi_2^*(a)(v,v')\ominus_\mathbb{Z} w(v,v'), &\text{if}~v\in V_1 \end{cases} \\
                            &= \begin{cases} \min\limits_{(v,v')\in E}  a(v')\ominus_\mathbb{Z} w(v,v'), &\text{if}~v\in V_0 \\
      \max\limits_{(v,v')\in E} a(v')\ominus_\mathbb{Z} w(v,v'), &\text{if}~v\in V_1 \end{cases} \\
      &= \mathcal{E}(a)(v)
      %%%%%%%%%%%%%%%%%%%%%%%%%%%%%%%%%
%    &= (\min_{\pi_1^0} \uplus \max_{\pi_1^1}) (a(v') \ominus_\mathbb{Z} w(v,v'))) \\
%    &= \begin{cases} \min\limits_{\pi_1^0 (v,v') = v} a(v') \ominus_\mathbb{Z} w(v,v'), &(v,v')\in E_0 \\
%      \max\limits_{\pi_1^1 (v,v') = v} a(v') \ominus_\mathbb{Z} w(v,v'), &(v,v')\in E_1 \end{cases} \\
%    &= \begin{cases} \min\limits_{(v,v')\in E} a(v') \ominus_\mathbb{Z} w(v,v'), &v\in V_0 \\
%      \max\limits_{(v,v')\in E} a(v') \ominus_\mathbb{Z} w(v,v'), &v\in V_1 \end{cases} \\
%    &= \mathcal{E}(a)(v)
\end{align*}
\end{proof}

Next we show non-expansiveness of the function $\mathcal{E}$ by
viewing it as the the composition of non-expansive functions. Using
prior knowledge (non-expansiveness of reindexing, $\min$ and $\max$
proven in \cite{bekp:fixpoint-theory-upside-down}), it suffices to
show that $\mathit{sub}_w$ is non-expansive.

\begin{lemma}
  The function $\mathit{sub}_w\colon K^E \to K^E$, defined via
  $\mathit{sub}_w (a)(e) = a(e) \ominus_\mathbb{Z} w(e)$ for
  $a\colon E\to K$, $e\in E$ and $w\colon E\to \mathbb{Z}$, is
  non-expansive.
\end{lemma}

\begin{proof}
  Let $w\colon E\to \mathbb{Z}$, $a,b\colon E\to K$. Without loss of
  generality, assume
  \[ \parallel \mathit{sub}_w(b) \ominus \mathit{sub}_w(a) \parallel =
    \mathit{sub}_w(b)(e) \ominus \mathit{sub}_w(a)(e)= (b(e)
    \ominus_\mathbb{Z} w(e)) \ominus (a(e)\ominus_\mathbb{Z} w(e))\]
  for some $e\in E$. We omit the trivial case, i.e.\ we assume
  $\parallel \mathit{sub}_w(b) \ominus \mathit{sub}_w(a) \parallel >
  0$. Thus $b(e) > a(e)$ has to hold by monotonicity.

   We make the following distinction of cases.
  \begin{enumerate}
  \item $w(e) \geq 0 \land a(e) \geq w(e)$: 
  \begin{align*}
   &(b(e) \ominus_\mathbb{Z} w(e)) \ominus (a(e)\ominus_\mathbb{Z}
   w(e)) \\
   &= (b(e) - w(e)) - (a(e) - w(e)) &&\mbox{[$b(e) > a(e)\geq w(e)$]} \\
   &= b(e) - a(e) \\
   &\leq \parallel b\ominus a\parallel 
    \end{align*}  
    
      \item $w(e) \geq 0 \land a(e) < w(e)$: 
  \begin{align*}
   &(b(e) \ominus_\mathbb{Z} w(e)) \ominus (a(e)\ominus_\mathbb{Z}
   w(e)) \\
   &= (b(e) \ominus w(e)) \ominus 0 &&\mbox{[$a(e) < w(e)$]} \\
   &\leq  b(e) - a(e) &&\mbox{[$a(e) < w(e),~a(e) < b(e)$]} \\
   &\leq \parallel b\ominus a\parallel 
    \end{align*}  
    
  \item $w(e) < 0 \land b(e) - w(e) \leq k$: 
  \begin{align*}
   &(b(e) \ominus_\mathbb{Z} w(e)) \ominus (a(e)\ominus_\mathbb{Z} w(e)) \\
   &= (b(e) - w(e)) - (a(e) - w(e)) &&\mbox{[$b(e) > a(e)$]} \\
   &= b(e) - a(e) \\
   &\leq \parallel b\ominus a\parallel 
    \end{align*} 
    
              \item $w(e) < 0 \land b(e) - w(e) > k$: 
  \begin{align*}
   &(b(e) \ominus_\mathbb{Z} w(e)) \ominus (a(e)\ominus_\mathbb{Z}
   w(e)) \\
   &= k - (a(e) \ominus_\mathbb{Z} w(e)) &&\mbox{[$a(e) \ominus_\mathbb{Z} w(e) \in K$]}\\
    &\leq k -  (a(e)\ominus_\mathbb{Z} (b(e)-k)) &&\mbox{[$w(e) < b(e)-k$]}  \\
    &= k - \min \{ \max \{ a(e) - b(e) + k ,0\} ,k\} \\
    &= k - \min \{ a(e)-b(e)+k , k\} &&\mbox{[$k \geq b(e)$]} \\
    &= k - (a(e)-b(e)+k) &&\mbox{[$b(e) > a(e)$]} \\
    &= b(e) -a(e) \\
    &\leq \parallel b\ominus a\parallel 
    \end{align*} 
  \end{enumerate}
\end{proof}
\end{toappendix}
Moreover $\mathcal{E}: K^V \to K^V$ can be proved to be non-expansive
  by showing that it can be expressed in term of basic functions which are known or easily shown to be non-expansive and exploiting the fact that non-expansiveness is preserved by composition (see Appendix~\ref{se:approximation-eg} \pb{for more details}).
 Thus both generalised strategy iteration approaches in \S\ref{se:GSI},
from below and from above, can be applied for determining
$\mu \mathcal{E}$, i.e., the solution of $\Gamma$.

% The theory in \S\ref{ss:upsidedown} can hence be used to derive the
% approximation $\mathcal{E}_\#^a$ for any $a\colon V\to K$ (details are
% in Appendix~\ref{se:approximation-eg}), which plays a basic role in the
% approach from above.

\begin{toappendix}
\begin{propositionrep}[non-expansiveness]
  The function $\mathcal{E} : K^V \to K^V$ is
  non-expansive.
\end{propositionrep}

\begin{proof}
  Immediate from the decomposition of $\mathcal{E}$ proved above and the fact that, as shown in~\cite{bekp:fixpoint-theory-upside-down}, non-expansiveness is preserved by composition.
\end{proof}

Next we determine the approximation $(\mathit{sub}_w)_\#^a$.

\begin{lemma}
  Given $w\colon E \to \mathbb{Z}$ and $a\colon E \to K$ the
  approximation
  $(\mathit{sub}_w)_\#^a \colon K^{[E]^a}\to
  K^{[E]^{\mathit{sub}_w(a)}}$ of $\mathit{sub}_w\colon K^E \to K^E$,
  is given by
  \[ (\mathit{sub}_w)_\#^a (E') = \{ e\in E' \mid 0< a(e) - w(e) \leq
    k \} \] for $E'\subseteq [E]^a$.
%[Dual proof: $(\mathit{sub}_w)^\#_a (E') = \{ e\in E' \mid 0\leq a(e) - w(e) < k \}$ for $E'\subseteq [E]_a$.]
\end{lemma}

\begin{proof}
  Note that here the minimal possible decrease is
  $\delta=1 \le \delta^a$, which we will use in the following

  Now:
  \begin{align*}
    (\mathit{sub}_w)_\#^{a,\delta}(E') &=
    \{ e\in [E]^{\mathit{sub}_w(a)}\mid \mathit{sub}_w(a)(e) \ominus \mathit{sub}_w(a\ominus \delta_{E'})(e)\ge \delta \} \\
    &= \{ e\in [E]^{\mathit{sub}_w(a)}\mid (a(e)\ominus_\mathbb{Z}
    w(e)) \ominus ((a\ominus \delta_{E'})(e)\ominus_\mathbb{Z}
    w(e))\ge \delta \}
  \end{align*}  
  First note, that whenever $e\notin E'$ then
  \[ (a(e)\ominus_\mathbb{Z} w(e)) \ominus ((a\ominus
    \delta_{E'})(e)\ominus_\mathbb{Z} w(e)) =(a(e)\ominus_\mathbb{Z}
    w(e)) \ominus (a(e)\ominus_\mathbb{Z} w(e)) = 0 < \delta. \] 
    
%    Now let $e\in E'\subseteq [E]^a$, i.e.\ $a(e) > 0$. We have $e\in  [E]^{\mathit{sub}_w(a)}$ iff $a(e)-w(e) > 0$. Next, whenever $a(e) - w(e) > k$ we obtain
%\[ (a(e)\ominus_\mathbb{Z} w(e)) \ominus ((a\ominus \delta_{E'})(e)\ominus_\mathbb{Z}
%    w(e)) =k - \min\{a(e)-\delta - w(e),k\} < k - \min\{ k-\delta,k\} \le \delta. \]

  Now let $e\in [E']^{\mathit{sub}_w(a)}$, i.e.\
  $a(e) - w(e) > 0$. Whenever $a(e) - w(e) > k$ we obtain
  \[ (a(e)\ominus_\mathbb{Z} w(e)) \ominus ((a\ominus
    \delta_{E'})(e)\ominus_\mathbb{Z} w(e)) = k - \min\{a(e)-\delta -
    w(e),k\} < k - \min\{ k-\delta,k\} \le \delta. \]
    
  If on the other hand $e\in E'$ with
  $0 < a(e) - w(e) \leq k$, we have:
  \[ (a(e)\ominus_\mathbb{Z} w(e)) \ominus ((a\ominus
    \delta_{E'})(e)\ominus_\mathbb{Z} w(e)) =(a(e)- w(e)) -
    (a(e)-\delta - w(e)) = \delta, \] since $E' \subseteq [E]^a$ and
  by choice of $\delta$, $a(e)-\delta -w(e) \geq 0$ holds.

  To summarize we obtain
  \[ (\mathit{sub}_w)_\#^a (E') = \{ e\in [E']^{\mathit{sub}_w(a)}
    \mid a(e) - w(e) \leq k \} = \{ e \in E' \mid 0 < a(e) - w(e) \leq
    k \}. \]
\end{proof}

Now we are able to derive $\mathcal{E}_\#^a$.

\begin{lemma}
  Let $V'\subseteq [V]^a$ then $v\in \mathcal{E}_\#^a (V')$ if
  $v\in [V]^{\mathcal{E}(a)}$ and
  \begin{itemize}
  \item whenever $v\in V_0$ there exists some $(v,v'')\in E$ with
    $\min\limits_{(v,v')\in E} a(v')\ominus_\mathbb{Z} w(v,v') =
    a(v'')\ominus_\mathbb{Z} w(v,v'')$,
    %0\sqsubset a(v''),~
    $0 < a(v'') - w(v,v'')
    \leq k$ and $v''\in V'$
  \item whenever $v\in V_1$: if $(v,v'')\in E$ with
    $\max\limits_{(v,v')\in E} a(v')\ominus_\mathbb{Z} w(v,v') =
    a(v'')\ominus_\mathbb{Z} w(v,v'')$ then
    $0 < a(v'') - w(v,v'') \leq k$ and $v''\in V'$
  \end{itemize}
  % \begin{align*}
  %   {\mathcal{E}'}_\#^a (V')= \{ v\in [V]^{\mathcal{E}'(a)} \mid
  %   &(v\in V_0 \land \text{ there exists some $(v,v'')\in E$ with
  % }\\
  %   &\min_{\pi_1^0(v,v') = v} a(v')\ominus' w(v,v') =  a(v'')\ominus' w(v,v'')\\
  %   &\text{ and }  0 < a(v'') \ominus' w(v,v'') \leq k,~v''\in V' ) \\
  %   \lor &(v\in V_1 \land \text{ if $(v,v'')\in E$ with }
  %   \max_{\pi_1^0(v,v') = v} a(v')\ominus' w(v,v') =  a(v'')\ominus' w(v,v'')\\
  %   &\qquad  \text{ then } 0 < a(v'') \ominus' w(v,v'') \leq k,~v''\in
  %   V') \}
%\end{align*}
\end{lemma}

\begin{proof}
  We have
  \begin{align*}
    \mathcal{E}_\#^a (V') = \{ v\in [V]^{\mathcal{E}(a)} \mid &(v\in V_0 \land \mathit{Min}_{{\mathit{sub}_w\circ \pi_2^*}(a)\mid_{E_0}} \cap {\mathit{sub}_w}_\#^{\pi_2^*(a)}(\pi_2^{-1}(V')) \neq \emptyset) \\
    \lor &(v\in V_1 \land \mathit{Max}_{{\mathit{sub}_w\circ \pi_2^*}(a)\mid_{E_1}} \subseteq {\mathit{sub}_w}_\#^{\pi_2^*(a)}(\pi_2^{-1}(V')))\} \\
    = \{ v\in [V]^{\mathcal{E}(a)} \mid &(v\in V_0 \land \mathit{Min}_{{\mathit{sub}_w\circ \pi_2^*}(a)\mid_{E_0}} \\
    &\cap \{ (v,v')\in E \mid 0 < a(v') - w(v,v') \leq k,~v'\in V'\} \neq \emptyset) \\
    \lor &(v\in V_1 \land \mathit{Max}_{{\mathit{sub}_w\circ \pi_2^*}(a)\mid_{E_1}} \\ &\subseteq  \{ (v,v')\in E \mid 0 < a(v') - w(v,v') \leq k,~v'\in V'\}\} \\
    = \{ v\in [V]^{\mathcal{E}(a)} \mid \quad &(v\in V_0 \land \text{ there exists some $(v,v'')\in E$ with }\\
    &\qquad\min_{(v,v')\in E} a(v')\ominus_\mathbb{Z} w(v,v') =  a(v'')\ominus_\mathbb{Z} w(v,v''),\\
    &\qquad 0 < a(v'') - w(v,v'') \leq k \text{ and } v''\in V' ) \\
    & \lor (v\in V_1 \land \text{ if $(v,v'')\in E$ with } \\
    & \qquad
    \max_{(v,v')\in E} a(v')\ominus_\mathbb{Z} w(v,v') =  a(v'')\ominus_\mathbb{Z} w(v,v''),\\
    &\qquad \text{ then } 0 < a(v'') - w(v,v'') \leq k \text{ and }
    v''\in V') \}
  \end{align*}
\end{proof}

On fixpoints this simplifies to
$v\in \mathcal{E}_\#^a (V')$ if $v\in [V]^a$ and
\begin{itemize}
\item $v\in V_0$: there exists some $(v,v')\in E$ with
  $a(v) = a(v')\ominus_\mathbb{Z} w(v,v')$, 
  $0 < a(v') - w(v,v') \leq
  k$ and $v'\in V'$
\item $v\in V_1$: if $(v,v')\in E$ with $a(v) = a(v')\ominus_\mathbb{Z} w(v,v')$
  then
  $0 < a(v') - w(v,v') \leq
  k$ and $v'\in V'$
\end{itemize}
The smallest propagation of~$1$ can always be used to skip fixpoints,
but in order to progress faster, a valid choice for
the decrease $\delta$ is
\begin{eqnarray*}
  \delta & = & \min \{ \{ a(v')- w(v,v')\mid a(v') >
  w(v,v'),~(v,v')\in E\} \\
  && \qquad \cup \{ a(v)- a(v') \mid a(v) >
  a(v'),~(v,v')\in E\} \cup \{\delta^a\} \}.
\end{eqnarray*}
\end{toappendix}

Observe that the algorithms do not only compute $\mu \mathcal{E}$, but
also provide an optimal strategy, for Player~0 when approaching from
above and for Player~1 when approaching from below. The second case is
of particular interest as it derives an optimal strategy for Player~1,
which is often not treated in the literature (we are only aware of~\cite{BrimChal:2010}).

We also remark that, when performing iteration from above or below, at
each iteration, once a strategy $C$ for Player~0 is fixed, we need to
compute $\mu \mathcal{E}_C$. This can be done via linear
programming, however it turns out that it is more efficient to use
some form of value iteration, due to finiteness of the MV-algebra
$\{0,\dots,k\}$.

Appendix~\ref{se:approximation-eg} \pb{also} spells out the
\pb{construction of the} approximation $\mathcal{E}_\#^a$ of the function $\mathcal{E}$ which
-- according to the theory in \S\ref{ss:upsidedown} -- can be used for
checking whether its least fixpoint has already been reached in
strategy iteration from above. In
Appendix~\ref{se:comparison-energy-games} we analyse known
algorithms for solving energy games and compare their runtime to both
kinds of strategy iteration.  While other algorithms might in some
cases have better runtimes, strategy iteration has the advantage of
providing the optimal strategy.
\begin{toappendix}
  
\subsection{Comparison to the state of art and runtime results (energy
  games)}
\label{se:comparison-energy-games}

In this section we briefly explain the transformation to energy games with finite values and depict two other algorithms to solve energy games. In the end we compare runtimes of our strategy algorithms to these known algorithms.

\paragraph*{Transformation into a game with finite values.}
We quickly detail the transformation from any energy game $\Gamma = (V_0,V_1,E,w)$ into a game $\Gamma' = (V_0',V_1',E',w')$ with finite  values presented in \cite{DBLP:conf/icalp/DorfmanKZ19}. The basic premise is that any play on $\Gamma'$ will not run in a negative cycle if Player~0 plays optimally. The transformation entails two simple steps:
\begin{enumerate}
\item Find cycles of negative weight exclusively consisting of states controlled by Player~1 and all states in $V_1$ which can reach any of these negative Player~1-cycles along a path consisting of only states controlled by Player~1. Remove all these states and any incoming and outgoing edges to these removed states. It is rather clear that any removed state requires an initial infinite energy.
\item Add a sink state $s$, the edge $(s,s)$ and incoming edges
  $(v,s)$ for all $v\in V_0$. The new edge weights are given by
  $w'(s,s) = 0$ and $w'(v,s) = - 2\cdot n \cdot W$. $s$ serves as an
  ``emergency'' exit for Player~0 which allows him to avoid negative
  cycles.
\end{enumerate}
It is easy to compute $\mu \bar{\mathcal{E}}$ (solution of $\Gamma$)
given $\mu \mathcal{E}'$ (solution to $\Gamma'$ with finite values
where $K=\{ 0,\dots, 3\cdot n\cdot W\}$). To this end
\cite{DBLP:conf/icalp/DorfmanKZ19} proved the following:
\[ \mu \bar{\mathcal{E}}(v) = \begin{cases} \mu \mathcal{E}'(v),~\text{if}~\mu \mathcal{E}'(v) < n\cdot W \\ \infty ,~\text{otherwise} \end{cases} \]

Both strategy iterations require a transformation if the solution of $\Gamma$ is not finite whereas the following two algorithms presented in this section do not.

\begin{example}
\label{ex:EG3a}
  The energy game $\Gamma$ to the left is transformed to an energy game
  with finite values $\Gamma'$ as follows:
 
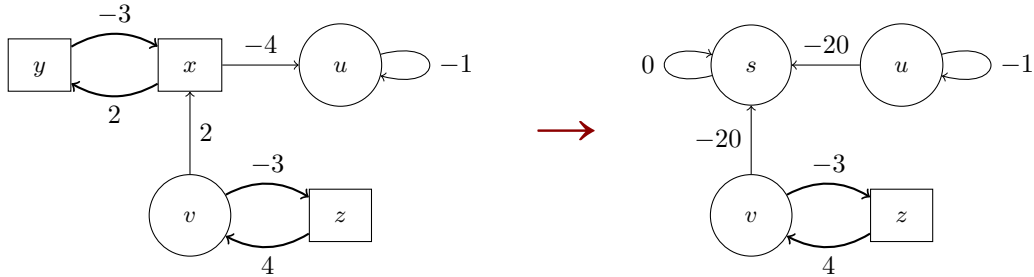
\begin{center}
   \begin{tikzpicture}
        \node (u) at (2,0) [circle,draw] {\begin{tabular}{c}
            $u$
          \end{tabular} };
        \node (v) at (0,-2) [circle,draw] {\begin{tabular}{c}
            $v$
          \end{tabular} };
        \node (x) at (0,0) [rectangle,draw] {\begin{tabular}{c}
            $x$
          \end{tabular} };
                  \node (y) at (-2,0) [rectangle,draw] {\begin{tabular}{c}
            $y$
          \end{tabular} };
        \node (z) at (2,-2) [rectangle,draw] {\begin{tabular}{c}
            $z$
          \end{tabular} };

\draw[->,   loop right] (u) to node [right]  {$-1$} (u);

        \draw[->,thick,bend left] (z) to node [below]  {$4$} (v);
        \draw[->,thick,bend left] (v) to node [above]  {$-3$} (z);
                        \draw[->,thick,bend left] (x) to node [below]  {$2$} (y);
        \draw[->,thick,bend left] (y) to node [above]  {$-3$} (x);
        \draw[->] (x) to node [above]  {$-4$} (u);
        \draw[->] (v) to node [right]  {$2$} (x);
        % \draw[->,   loop] (y) to node [above]  {$\frac{1}{4}$} (y);
%        \draw[->,thick, bend left] (x) to node [below]  {} (z);
%        \draw[->,thick, bend left] (z) to node [below]  {} (x);
        \node (arrow) at (5,-1) [] {\begin{tabular}{c}
            {\Huge {\red $\rightarrow$}}
          \end{tabular} };
      \end{tikzpicture}
      \begin{tikzpicture}
        \node (u) at (2,0) [circle,draw] {\begin{tabular}{c}
            $u$
          \end{tabular} };
        \node (s) at (0,0) [circle,draw] {\begin{tabular}{c}
            $s$
          \end{tabular} };
        \node (v) at (0,-2) [circle,draw] {\begin{tabular}{c}
            $v$
          \end{tabular} };

        \node (z) at (2,-2) [rectangle,draw] {\begin{tabular}{c}
            $z$
          \end{tabular} };
        
\draw[->,   loop left] (s) to node [left]  {$0$} (s);
\draw[->,   loop right] (u) to node [right]  {$-1$} (u);
\draw[->] (u) to node [above]  {$-20$} (s);
\draw[->] (v) to node [left]  {$-20$} (s);
        \draw[->,thick,bend left] (z) to node [below]  {$4$} (v);
        \draw[->,thick,bend left] (v) to node [above]  {$-3$} (z);

        % \draw[->,   loop] (y) to node [above]  {$\frac{1}{4}$} (y);
%        \draw[->,thick, bend left] (x) to node [below]  {} (z);
%        \draw[->,thick, bend left] (z) to node [below]  {} (x);
      \end{tikzpicture}
  \end{center}
\end{example}

\paragraph*{Kleene iteration.}

The simplest algorithm to solve energy games entails a simple Kleene
iteration for a fixpoint function
$\hat{\mathcal{E}}\colon \{ 0,\dots,k,\top\}^V \to \{
0,\dots,k,\top\}^V$ (where we introduce a top elements as discussed at
the end of \S\ref{se:eg-intro}).
In particular we define $g^{(0)} = 0$ (i.e., $g^{(0)}(v)=0$ for all $v \in V$),
$g^{(i+1)}= \hat{\mathcal{E}}(g^{(i)})$ until
$g^{(n)} = g^{(n+1)} = \mu \hat{\mathcal{E}}$.

Note that $\{ 0,\dots,k, \top\}$ is a complete lattice, thus a
transformation of $\Gamma$ to $\Gamma'$ with finite values is not
required for this algorithm.

\paragraph*{Value iteration.}

Next we briefly discuss a value iteration developed by
\cite{DBLP:journals/fmsd/BrimCDGR11} that resembles a worklist
algorithm used in the context of dataflow analysis. The algorithm
starts with a value function $g=0$ and computes a list $L$ of invalid states
w.r.t.\ $g$, i.e.
\begin{itemize}
\item $v\in V_0$ invalid iff $w(v,v') +g(v)-g(v') < 0$ for all $(v,v')\in E$
\item $v\in V_1$ invalid iff $w(v,v') +g(v)-g(v') < 0$ for some $(v,v')\in E$
\end{itemize}
Now, we iteratively pick any state $v\in L$ and increase $g(v)$ until
$v$ is valid. This may produce new invalid states which need to be
added to $L$. The algorithm terminates when there are no more invalid
states, i.e.\ $L=\emptyset$, and we obtain
$g = \mu \bar{\mathcal{E}}$.

\begin{example}
  \label{ex:EG3}
  We continue with Example~\ref{ex:EG3a} and solve the transformed
  energy game via the presented algorithms which compute
  $\mu \mathcal{E}'$:
  \begin{itemize}
  \item \emph{Kleene iteration:} For the given example Kleene
    iteration requires $20$ steps, since the value of $u$ is increased
    by $1$ in each iteration and the solution for $u$ is in fact $20$.
    Only after 20 iteration, $u$ ``realises'' that it must choose $s$
    as its successor and thus the least fixpoint is reached.
  \item \emph{Value iteration:} Initially $L=\{ u,v\}$. Now, we choose
    $v\in L$ and make it valid, thus $g(v) = 3$. Now, $L=\{ u \}$
    since no other state is added. Next, we increase $g(u)$ to $20$ to
    make $u$ valid and the iteration terminates since $L=\emptyset$.
\item \emph{Strategy iteration for Player~0:} Assume $C_0 (u) = C_0
  (v)  = s$. Now, we can solve the linear program (or determine the
  solution by value iteration):
  \begin{align*}
    &\min ~g(u)+g(v)+g(z)+g(s) \\
    g(u) &\geq g(s) - (-20),~g(v) \geq g(s) - (-20),~g(s) = 0,~g(z) \geq g(v) - 4, \\
    0 &\leq g(u),g(v),g(z),g(s) .
  \end{align*}
  This yields $g^{(0)} (s) = 0$, $g^{(0)}(u) = g^{(0)}(v) = 20$ and
  $g^{(0)}(z) = 16$. Next, the strategy is updated and we obtain
  $C_1(u) = s$ and $C_1 (v) = z$. Now we have optimal strategies and
  we reach $\mu \mathcal{E} = \mu \mathcal{E}_{C_1}$ (e.g., by solving
  the resulting linear program).

\item \emph{Strategy iteration for Player~1:} It is evident that there
  is only one strategy for Player~1, i.e.\ $C_0 (z) = v$. Thus, we
  immediately obtain the solution
  $ \mu \mathcal{E}=\mu \mathcal{E}_{C_0}$ via Kleene or value
  iteration (as
  above). %Obviously $\sigma_1$ is the (only) optimal strategy for Player~1.
  \end{itemize}
\end{example}

\paragraph*{Runtime results.}

We now compare runtimes for each algorithm by generating random games
using the Erd{\H o}s\text{-}R{\'e}nyi random graph model. The number
of states in our system is denoted by $n$, while $p$ represents the
probability that any edge $(u,v)$ exists, i.e.\ given states $u$ and
$v$, the probability that the edge $(u,v)$ exists is given by
$p$.
%$p$ represents the probability for each edge $E(i,j)$ to exist, i.e.\ for $p=1$ every state has each other state as his successor. 
We guarantee at least one outgoing edge for each state.
If an edge exists, its weight is some uniformly random integer in
$[-W,W]$ where $W$ is the maximum edge weight.

Our runtime-tables use the following abbreviations:
\begin{itemize}
\item \textbf{TF}: transformation of $\Gamma$ to $\Gamma'$
\item \textbf{KLE}: Kleene iteration %on $\Gamma'$
\item \textbf{VI}: Value Iteration %on $\Gamma'$
\item \textbf{SI0}: Strategy Iteration for Player 0 (iteration from above)% - we use Kleene iteration instead of linear programming to compute $\mu \mathcal{E}'_{\sigma^{(i)}}$
\item \textbf{SI1}: Strategy Iteration for Player 1 (iteration from below)
%\item \textbf{L0}: the algorithm by \cite{DBLP:conf/icalp/DorfmanKZ19} on $\Gamma'$
\end{itemize}
We use value iteration to compute the $\mu \mathcal{E}_{C_i}$ in both strategy iterations given any strategy $C_i$. This is more efficient than using linear programming in \textbf{SI0}. 

Cumulative runtimes (in seconds) are given for 1000 random runs for $W=n$ and $p=\frac{2}{n}$, where each state randomly belongs to Player~0 or Player~1. A deviation with respect to $W$ and $p$ barely impacts the runtime comparison. 

First, we examine random systems of exclusively finite values:
\begin{center}
\scalebox{1.2}{
\begin{tabular}{|c||c|c|c|c|}
\hline 
 &  \textbf{KLE} & \textbf{VI}   & \textbf{SI1} & \textbf{SI0}   \\ 
\hline \hline
$n=20$  & 0.04 & 0.02 & 0.07 & 0.27  \\ 
\hline 
$n=40$ & 0.13 & 0.05 & 0.19 & 1.88   \\ 
\hline 
$n=80$ &  0.59 &0.17 & 0.74 & 14.79   \\ 
\hline 
\end{tabular} 
}
\end{center}
It is rather clear that approaching from above (\textbf{SI0}) does not seem fruitful since values are usually rather small. Note that \textbf{VI} performs the best.

Next, we examine random systems where infinite values are allowed. Hence, we need to transform these systems to systems with finite values (\textbf{TF}) beforehand and then apply our algorithms to these reduced system. We note that around every second state requires an infinite initial energy in the original systems:
\begin{center}
\scalebox{1.2}{
\begin{tabular}{|c||c|c|c|c|c|}
\hline 
 & \textbf{TF}&  \textbf{KLE} & \textbf{VI}   & \textbf{SI1} & \textbf{SI0}   \\ 
\hline \hline
$n=20$ & 0.05 & 0.33 & 0.19 & 0.48 & 0.1  \\ 
\hline 
$n=40$ & 0.19& 2.6 & 0.79 & 2.5 & 0.3  \\ 
\hline 
$n=80$& 1.07 & 26.7 &4.96 & 19.41 & 0.98   \\ 
\hline 
\end{tabular} 
}
\end{center}
It is not rare that a handful of states attain a value of around
$2\cdot n \cdot W$. Thus, \textbf{SI0} is rather efficient in this
instance when choosing the sink state $s$ as the initial successor for
each state in $V_0'$ (this strategy also guarantees finite
values). Here, \textbf{SI0} is very competitive, even compared to
\textbf{VI}. We however note that \textbf{SI1} is the only algorithm
which produces also an optimal strategy for Player~1.

\end{toappendix}

\section{Application: behavioural metrics for probabilistic automata}
\label{se:BM}

In this section we show how our technique can be used to compute behavioural distances
over probabilistic automata. After introducing the necessary notions, we provide a min-decomposition of the corresponding function. The algorithm that we obtain by instantiating our generalised strategy iteration from above using such min-decomposition can be seen to be essentially the same as the one presented in~\cite{bblmtb:prob-bisim-distance-automata-journal}.

\begin{definition}[probabilistic automaton]
  A \emph{probabilistic automaton} (PA) is a tuple
  $\mathcal{A}=(S,L,\delta,\ell)$ consisting of a nonempty finite set $S$ of
  states, a finite set of labels $L$, a successor function
  $\delta : S \to \Powf{\mathcal{D}(S)}$ and a labeling function
  $\ell \colon S\to L$. 
\end{definition}

The idea is that from a state $s$, one can non-deterministically move
to one of the probability distributions in $\delta(s)$.

The behavioural distance function is defined by combining Hausdorff
and Kantorovich liftings for the nondeterministic and probabilistic
parts, respectively. Recall that the \emph{Kantorovich lifting}
\cite{v:optimal-transport} $K:  [0,1]^{Y \times Y} \to \interval{0}{1}^{\mathcal{D}(Y) \times \mathcal{D}(Y) }$ transforming a pseudometric $d$ on $Y$ to a pseudometric on $\mathcal{D}(Y)$ is defined, for
$\beta, \beta' \in \mathcal{D}(Y)$, by

\begin{center}
  $K(d)(\beta,\beta') = \min_{\omega \in \Omega(\beta,\beta')} \sum_{y,y' \in Y}
  d(y,y') \cdot \omega(y,y')$,
\end{center}
where $\Omega(\beta,\beta')$ is the set of probabilistic couplings of $\beta,\beta'$:
\begin{center}
  $\Omega(\beta,\beta') = \{ \omega \in \mathcal{D}(Y \times Y) \mid
  \forall y,y' \in Y\colon \sum_{x' \in Y}\omega(y,x')
  = \beta(y) \mathop{\land} \sum_{x \in Y}\omega(x,y') = \beta'(y')
  \}$
\end{center}
Actually, the minimum is reached in one of the finitely many vertices
of the polytope $\Omega(\beta,\beta')$, a set which we denote by
$\Omega_V(\beta,\beta')$. The \emph{Hausdorff lifting} $H:  [0,1]^{Y \times Y} \to \interval{0}{1}^{\Pow{Y} \times \Pow{Y} }$
(in the variant
of \cite{m:wasserstein}) is defined, for $X, X' \in \Pow{Y}$, by
\begin{center}
  $H(d)(X,X') = \min_{R \in \mathcal{R}(X,X')} \max_{(x,x') \in R}
  d(x,x')$,
\end{center}
with $\mathcal{R}(X,X') = \{ R \in \Pow{Y \times Y} \mid \pi_1(R) = X \land \pi_2(R)=X'\}$ the set-couplings of $X,X'$~\cite{m:wasserstein}.

The rough idea is the following.  If two states $s$ and $t$ have
different labels they are at distance $1$. Otherwise, in order to
compute their distance one has find a ``best match'' between the
outgoing transitions of such states, i.e., a set coupling as those
considered in the Hausdorff lifting $H$. In turn, since, transitions
are probabilistic, matching transitions means finding an optimal
probabilistic coupling, as done by the Kantorovich lifting $K$, which
is intuitively  the best transport plan balancing the
``supply'' $\beta$ and the ``demand'' $\beta'$. In this way the
distance of $s$ and $t$ is expressed in terms of the distance of the
states they can reach, hence, formally, behavioural distance is
characterised as a least fixpoint.

\begin{definition}[behavioural distance]
  Let $\mathcal{A}=(S,L,\delta,\ell)$ be a PA. The
  behavioural distance on $\mathcal{A}$ is the least fixpoint of
  % the function
  $\mathcal{M} : [0,1]^{S\times S} \to [0,1]^{S\times S}$ defined, for
  $d\in [0,1]^{S\times S}$ and $s,t\in S$, by
  $\mathcal{M}(d)(s,t) = H(K(d))(\delta(s),\delta(t))$ if
  $\ell(s) = \ell(t)$ and $\mathcal{M}(d)(s,t) = 1$, otherwise.
\end{definition}

%\begin{toappendix}
%  \label{app:PA}

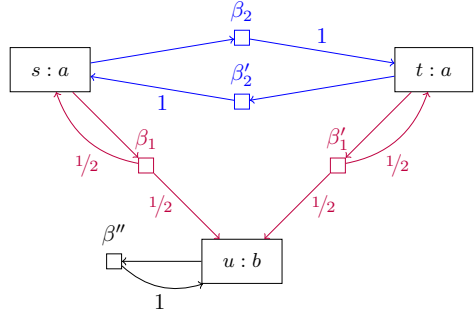
\begin{figure}
    \begin{center}
    \scalebox{0.85}{
    \begin{tikzpicture}	
      \node (s) at (0,0) [rectangle,draw] {\begin{tabular}{c}
          $s:a$
        \end{tabular} };
        
      \node (t) at (6,0) [rectangle,draw] {\begin{tabular}{c}
          $t:a$
        \end{tabular} };
        
      \node (u) at (3,-3) [rectangle,draw] {\begin{tabular}{c}
          $u:b$
        \end{tabular} };
        
      \node (b2) at (3,0.5) [color=blue, draw, label=\textcolor{blue}{$\beta_2$}] {};
      \node (b1) at (1.5,-1.5) [color=purple, draw, label=\textcolor{purple}{$\beta_1$}] {};
      \node (b2p) at (3,-0.5) [color=blue, draw, label=\textcolor{blue}{$\beta_2'$}] {};
      \node (b1p) at (4.5,-1.5) [color=purple, draw, label=\textcolor{purple}{$\beta_1'$}] {};
        
      \node (bpp) at (1,-3) [draw, label=$\beta''$] {};
        
      \draw[->,   blue] (s) to node [below]  {} (b2);      
      \draw[->,   blue] (b2) to node [above]  {1} (t); 
      
      \draw[->] (u) to node [above]  {} (bpp);      
      \draw[->,bend right] (bpp) to node [below]  {1} (u); 
      
      \draw[->,   blue] (t) to node [above]  {} (b2p);      
      \draw[->,   blue] (b2p) to node [below]  {1} (s);   
        
      \draw[->,   purple] (s) to node [right]  {} (b1);      
      \draw[->,   purple] (b1) to node [left=0.1]  {$\onehalf$} (u);  
      \draw[->,   purple,bend left] (b1) to node [below=0.1]  {$\onehalf$} (s); 
      
      \draw[->,   purple] (t) to node [left]  {} (b1p);      
      \draw[->,   purple] (b1p) to node [right=0.1]  {$\onehalf$} (u);  
      \draw[->,   purple,bend right] (b1p) to node [below=0.1]  {$\onehalf$} (t);    

    \end{tikzpicture}
    }
  \end{center}

  \caption{A probabilistic automaton.}
  \label{fi:probabilistic}
\end{figure}

\begin{example}
  Consider the probabilistic automaton in Fig.~\ref{fi:probabilistic} with state space
  $Y=\{s,t,u\}$, labels $\ell(s) = \ell (t) = a$ and $\ell(u)=b$ and
  probability distributions
  $\beta_1,\beta_2,\beta'_1,\beta'_2,\beta''$ as indicated. For
  instance, from state $s$, there are two possible transitions
  $\beta_1$ which with probability $\onehalf$ goes to $u$ and with
  probability $\onehalf$ stays in $s$, and $\beta_2$ which goes to $t$ with
  probability $1$.

  In order to explain how function $\mathcal{M}$, resulting from the
  combination of Hausdorff and Kantorovich lifting, works, let us
  consider the pseudometric $d(s,t) = \nicefrac{1}{2}$,
  $d(s,u) = d(t,u) = 1$. This is not the least fixpoint, since the
  distance of states $s,t$ is clearly $0$ as the two states exihibit
  the same behaviour.

  We now illustrate how to compute $\mathcal{M}(d)(s,t)$. We obtain
  $\mathcal{M}(d)(s,u) = \mathcal{M}(d)(t,u) = 1$ and, since
  $\ell(s) = \ell(t) = a$, we have
  \[
    \mathcal{M}(d)(s,t) = H(K(d))(\delta(s),\delta(t)).
  \]
  where $\delta(s) = \{ \beta_1, \beta_2\}$ and
  $\delta(t) = \{ \beta_1', \beta_2'\}$.

  It is relatively straightforward to see that the vertices of the
  coupling polytope $\Omega(\beta_1,\beta'_1)$ are
  $\Omega_V(\beta_1,\beta'_1) = \{ \omega_1,\omega_2\}$ with
  \[
    \omega_1(s,t) = \nicefrac{1}{2},~\omega_1(u,u) = \nicefrac{1}{2}
    \quad \text{ and } \quad
    \omega_2(s,u) = \nicefrac{1}{2},~\omega_2(u,t) = \nicefrac{1}{2}
  \]
  and $\omega_i (x,y) = 0$, $i \in \{1,2\}$, for every other pair
  $(x,y)\in Y\times Y$.
  Then the Kantorovich lifting is determined as follows:
  \[
    K(d)(\beta_1,\beta'_1)  = \min \{ \sum_{x,y\in S} d(x,y) \cdot \omega_1(x,y), \sum_{x,y\in S} d(x,y) \cdot \omega_2(x,y)\} = \min\{ \nicefrac{1}{4}, 1 \} = \nicefrac{1}{4}.
  \]

  Similarly we can obtain $K(d)(\beta_1,\beta'_2)= \nicefrac{1}{2}$,
  $K(d)(\beta_2,\beta'_1)= \nicefrac{1}{2}$, $K(d)(\beta_2,\beta'_2)= \nicefrac{1}{4}$.

  In order to conclude the computation via the Hausdorff lifting,
  note that the minimal set-couplings of
  $\delta(s)=\{\beta_1,\beta_2\}$ and
  $\delta(t)=\{\beta'_1,\beta'_2\}$ are
  \[
    R_1 = \{ (\beta_1,\beta'_1),(\beta_2,\beta'_2)\}
    \qquad
    R_2 = \{ (\beta_1,\beta'_2),(\beta_2,\beta'_1)\}
  \]
  and any other set-coupling includes $R_1$ or $R_2$.
  Then we obtain
  \begin{align*}
    \mathcal{M}(d)(s,t)
    &={H}({K}(d))(\delta(s),\delta(t)) \\
    % = \min_{R\in \mathcal{R}(X,X')} \max_{(x,x')\in R} {K}(d)(x,x') 
    &= \min \{ \max_{(x,x')\in R_1} {K}(d)(x,x'), \max_{(x,x')\in R_2} {K}(d)(x,x')\} \\
    &= \min \{ \max \{ {K}(d)(\beta_1,\beta'_1),{K}(d)(\beta_2,\beta'_2)\}, \max \{ {K}(d)(\beta_1,\beta'_2),{K}(d)(\beta_2,\beta'_1)\} \}\\
    &= \min \{ \max \{ \nicefrac{1}{4},\nicefrac{1}{4}\}, \max \{ \nicefrac{1}{2},\nicefrac{1}{2}\} \} 
      = \min \{ \nicefrac{1}{4},\nicefrac{1}{2}\} = \nicefrac{1}{4}.
  \end{align*}
\end{example}
%\end{toappendix}

In order to cast this problem in our framework, we identify a suitable
min-decomposition of $\mathcal{M}$. Observe that, for
$d\in [0,1]^{S\times S}$ and $s,t\in S$ such that $\ell(s)=\ell(t)$,
expanding the definitions of the liftings and taking advantage of
complete distributivity, we have
\begin{align*}
  \mathcal{M}(d) (s,t) &
  = \min\limits_{R \in \mathcal{R}(\delta(s),\delta(t))} \max\limits_{(\beta,\beta')\in R} \min\limits_{\omega \in \Omega_V(\beta,\beta')} \sum\limits_{u,v\in S} d(u,v) \cdot \omega(u,v)\\
  & = \min_{R\in \mathcal{R}(\delta(s),\delta(t))} \min_{f\in F_R}
  \max_{(\beta,\beta')\in R} \sum\limits_{u,v\in S} d(u,v) \cdot
  f(\beta,\beta')(u,v)% \\
  % &= \min_{R\in \mathcal{R}(\delta(s),\delta(t)),f\in F}
  % \max_{(\beta,\beta')\in R} \sum_{u,v\in S} d(u,v) \cdot f(\beta,\beta')(u,v)
\end{align*}
where
$F_R=\{ f\colon R \to \mathcal{D}(S\times S)\mid f(\beta,\beta') \in
\Omega_V(\beta,\beta') \mbox{ for } (\beta,\beta') \in R\}$, which is a finite
set.

% and $u\colon \Pow{S\times S} \to \Pow{S} \times \Pow{S}$ is defined as $u(C) = (\pi_1(C),\pi_2(C))$ where $\pi_1$ and $\pi_2$
%are the projections $\pi_i\colon S\times S \to S$ and $\pi_i[C] = \{ \pi_i(c)\mid c\in C\}$.

%\[ \mathcal{M}(a) (x,y) = \min_{u_1(C) =(\delta(x),\delta(y))} \max_{(\beta,\beta')\in C} \min_{c\in \Omega(\beta,\beta')} \sum_{u,v\in X} a(u,v) \cdot c(u,v).\]
%Here, $\Omega(\beta,\beta')$ denotes the set of (Wasserstein-)couplings for $\beta,\beta'$.

We can thus define a min-decomposition $H_{\min}$ for $\mathcal{M}$ (see Definition~\ref{de:min-decomposition}) such that 
$\mathcal{M}(d) (s,t) = \min_{h\in H_{\min}(s,t)} h(d)$ for all $s, t \in S$.

\begin{definition}[min-decomposition of $\mathcal{M}$]
  \label{de:pa-min-decomposition}
  Let $\mathcal{A}=(S,L,\delta,\ell)$ be a PA. We
  denote by $H_{\min}$ the min-decomposition of $\mathcal{M}$ defined
  as follows. For $s, t \in S$ such that $\ell(s)=\ell(t)$, we let
  $H_{\min}(s,t) = \{ h_{R,f} \mid R\in
  \mathcal{R}(\delta(s),\delta(t)), f\in F_R \}$, with
  $h_{R,f} : [0,1]^{S \times S} \to [0,1]$ defined as
  \begin{center}
    $h_{R,f}(d) = \max_{(\beta,\beta') \in R} \sum\limits_{u,v\in S} d(u,v) \cdot f(\beta,\beta')(u,v)$.
  \end{center}
  If instead $\ell(s) \neq \ell(t)$, we let
  $H_{\min}(s,t) = \{ h_1 \}$ where $h_1(d) = 1$ for all $d$.
\end{definition}

A strategy $C$ in $H_{\min}$ maps each pair of states $s,t\in S$ to a function in $H_{\min}(s,t)$, that is
\begin{itemize}
\item if $\ell (s) \neq \ell(t)$, to the unique element $h_1 \in H_{\min}(s,t)$;
\item if $\ell (s) = \ell(t)$ to some $h_{R,f} \in H_{\min}(s,t)$,
  where $R \in \mathcal{R}(\delta (s),\delta(t))$ is a set-coupling
  and $f\in F_R$.
\end{itemize}

% We can thus instantiate
The decomposition above can be used to deduce that $\mathcal{M}$ is non-expansive and thus we can safely
instantiate the algorithm in
Fig.~\ref{fi:alg-min-fix-above} to compute the least fixpoint from
above.  The
resulting algorithm is quite similar to the one specifically developed
for PAs in~\cite{bblmtb:prob-bisim-distance-automata-journal}. In particular, it can
be seen that, apart from the different presentation, a strategy $C$
corresponds to what~\cite{bblmtb:prob-bisim-distance-automata-journal}
refers to as a \emph{coupling structure}.
In addition, the step in item~(\ref{fi:alg-min-fix-above:2c}) of the
algorithm (see Fig.~\ref{fi:alg-min-fix-above}) is analogous to that
in~\cite{bblmtb:prob-bisim-distance-automata-journal}. In fact, in
order to check whether the fixpoint obtained with the current strategy
$C_i$, i.e.\ $\mu \mathcal{M}_{C_i}$, is the least fixpoint of
$\mathcal{M}$, one considers the approximation
$\mathcal{M}_\#^{\mu \mathcal{M}_{C_i}}$ and checks whether its
greatest fixpoint is empty. Recalling that the post-fixpoints of
$\mathcal{M}_\#^{\mu \mathcal{M}_{C_i}}$ have been shown
in~\cite{bekp:fixpoint-theory-upside-down} to be the self-closed
relations of~\cite{bblmtb:prob-bisim-distance-automata-journal}, one
derives that verifying the emptiness of the greatest fixpoint of
$\mathcal{M}_\#^{\mu \mathcal{M}_{C_i}}$ corresponds exactly to
checking whether the largest self-closed relation is empty (see 
Appendix~\ref{app:PA} for more details).

\medskip

\begin{toappendix}
\label{app:PA}
The following lemma is helpful in the instantiation of the algorithm, when we need to construct a new strategy.

\begin{lemma}
  \label{le:improve-pa}
  Let $\mathcal{A}=(S,L,\delta,\ell)$ be a PA
  and let $H_{\min}$ be the min-decomposition of $\mathcal{M}$ in
  Definition~\ref{de:pa-min-decomposition}. Given a strategy $C$ in $H_{min}$
  and $d : Y \times Y \to [0,1]$, a strategy $C'(y) =\arg\min_{h\in H_{\min}(y)} h(d)$
  can be defined as follows: for
  $(s,t) \in S \times S$
  \begin{itemize}
  \item if $\ell (s) \neq \ell(t)$ then $C'(s,t) = C(s,t)$
  \item if $\ell (s) = \ell(t)$ then $C'(s,t) = h_{R',f'}$ where
    \begin{align*}
      R'
      & = {\arg\min}_{R\in \mathcal{R}(\delta(s),\delta(t))} \max_{(\beta,\beta')\in R} K(d)(\beta,\beta')
    \end{align*}
    and for $(\beta,\beta')\in R'$:
    \[
      f'(\beta,\beta') = \arg\min_{\omega\in \Omega_V(\beta,\beta')}
      \sum_{u,v\in S} d(u,v) \cdot \omega(u,v).
    \]
  \end{itemize}
\end{lemma}

\begin{proof}
  Let $(s,t) \in S \times S$. If $\ell (s) \neq \ell(t)$ then
  $H_{\min}(s,t) = \{ h_1\}$, hence the only possible choice is
  $C'(s,t) = h_1 = C(s,t)$.

  If instead $\ell (s) = \ell(t)$ then $C'(s,t) = h_{R',f'}$ is chosen
  in a way that minimises
  \begin{align}
    \label{eq:f}
    h_{R',f'}(d) = \max_{(\beta,\beta')\in R'}
    \sum_{u,v\in S} d(u,v) \cdot f'(\beta,\beta')(u,v)
  \end{align}
  
  In order to minimise the value above, whatever $R'$ will be, for all
  $(\beta,\beta')\in R'$ the choice of $f'(\beta,\beta')$ should minimise
  $\sum_{u,v\in S} d(u,v) \cdot \omega(u,v)$.
  Formally, we can define
  $F : \mathcal{D}(S)\times \mathcal{D}(S)\to \mathcal{D}(S\times S)$
  as
  \[
    F(\beta,\beta') = \arg\min_{\omega\in \Omega_V(\beta,\beta')}
    \sum_{u,v\in S} d(u,v) \cdot \omega(u,v).
  \]
  
  Then the set-coupling
  $R'$ can be
  \begin{align*}
    R'
    &
    = {\arg\min}_{R\in \mathcal{R}(\delta(s),\delta(t))} \max_{(\beta,\beta')\in R}  \sum_{u,v\in S} d(u,v) \cdot F'(\beta,\beta')(u,v)\\
    &
      = {\arg\min}_{R\in \mathcal{R}(\delta(s),\delta(t))} \max_{(\beta,\beta')\in R} \min_{\omega \in \Omega_V(\beta,\beta')} \sum_{u,v\in S} d(u,v) \cdot \omega(u,v)\\
    &  = {\arg\min}_{R\in \mathcal{R}(\delta(s),\delta(t))} \max_{(\beta,\beta')\in R} K(d)(\beta,\beta')
  \end{align*}
  and finally we can define $f' = F_{|R'}$, i.e., explicitly, for all $(\beta, \beta') \in R'$
    \[
    f'(\beta,\beta') = \arg\min_{\omega\in \Omega_V(\beta,\beta')}
    \sum_{u,v\in S} d(u,v) \cdot \omega(u,v).
  \]
  as desired.
\end{proof}

The algorithm starts by fixing a strategy $C_0$
(item~(\ref{fi:alg-min-fix-above:1})). Then, at each iteration, if
$\minimp{C_{i}} \neq \emptyset$ (item~(\ref{fi:alg-min-fix-above:2b})
which by Lemma~\ref{lem:iff-above}, can be checked by verifying if
$\mathcal{M}(\mu \mathcal{M}_{C_i}) \sqsubset \mu \mathcal{M}_{C_i})$,
we consider a new strategy $C_{i+1} \in \minimp{C_{i}}$. According to
Remark~\ref{rem:min-imp}, this can be defined as follows: for
$(s,t) \in S \times S$
\begin{itemize}
\item if $\ell (s) \neq \ell(t)$ then $C_{i+1}(s,t) = C_i(s,t)$
\item if $\ell (s) = \ell(t)$ then $C_{i+1}(s,t) = h_{R',f'}$ chosen
  in a way that minimises $h_{R',f'}(\mu
  \mathcal{M}_{C_i})$. Concretely (see Lemma~\ref{le:improve-pa}), one can define
  \begin{align*}
    R'
    & = {\arg\min}_{R\in \mathcal{R}(\delta(s),\delta(t))} \max_{(\beta,\beta')\in R} K(\mu \mathcal{M}_{C_i})(\beta,\beta')
\end{align*}
and for $(\beta,\beta')\in R'$:
  \[
    f'(\beta,\beta') = \arg\min_{\omega\in \Omega_V(\beta,\beta')}
    \sum_{u,v\in S} \mu \mathcal{M}_{C_i}(u,v) \cdot \omega(u,v).
  \]
\end{itemize}

If instead, $\minimp{C_{i}} = \emptyset$
(item~(\ref{fi:alg-min-fix-above:2c})) and thus, by
Lemma~\ref{lem:iff-above}, $\mu \mathcal{M}_{C_i}$ is a fixpoint of
$\mathcal{M}$, we check whether it is the least fixpoint by verifying
if $\nu \mathcal{M}_\#^{\mu \mathcal{M}_{C_i}} =\emptyset$
(Theorem~\ref{th:fixpoint-sound-compl}), and in case it is not, we use
Lemma~\ref{lem:fp-increase} to determine a pre-fixpoint
$a \sqsubset \mu \mathcal{M}_{C_i}$, which is then used to obtain
$C_{i+1}$.
Everything works since $\mathcal{M}$ is monotone and non-expansive
(its approximation $\mathcal{M}_\#^d$ is spelled out in~\cite{bekp:fixpoint-theory-upside-down}). Furthermore
$\mu \mathcal{M}_{C_i}$ is again obtained by linear programming,
similar to the case of simple stochastic games.

\medskip

The resulting algorithm is quite similar to the one specifically
developed for PAs
in~\cite{bblmtb:prob-bisim-distance-automata-journal}. In particular, apart from the different presentation, a strategy $C$ corresponds
to what~\cite{bblmtb:prob-bisim-distance-automata-journal} refers as a \emph{coupling structure}.

More in detail, in our case, a strategy $C$ in $H_{\min}$ maps each
pair of states $s,t\in S$ with $\ell (s) = \ell(t)$ to some
$h_{R,f} \in H_{\min}(s,t)$, where
$R \in \mathcal{R}(\delta (s),\delta(t))$ is a set-coupling and
$f\in F_R$ maps each $(\beta,\beta') \in R$ to a probabilistic
coupling. Note that the choice of the probabilistic couplings is
``local'', i.e., we could have different pairs of states
$(s,t), (u,v) \in S \times S$ and $C(s,t) = (R, f)$,
$C(u,v) = (R',f')$, with $(\beta,\beta') \in R \cap R'$ and
$f(\beta,\beta') \neq f'(\beta,\beta')$.
%
% : if two states $s,u$ both happen to
% transition to the same probability distribution $\beta$ and $t,v$ both
% to $\beta'$, then a strategy $C$ could choose different probabilistic
% couplings for $\beta,\beta'$ relative to state pairs
% $(s,t), (u,v)$
However, it is easy to see that
we can assume (and it is computationally convenient to do so) that the
choice of the probabilistic coupling is actually ``global'', i.e.,
that for a strategy $C$, there is a (partial) function
$F : \mathcal{D}(S)\times \mathcal{D}(S)\to \mathcal{D}(S\times S)$
such that for each $(s,t) \in S \times S$ we have
$C(s,t) = (R, F_{|R})$.  In this view, a strategy $C$ can be
identified with a pair $(\rho, F)$, where $\rho$ gives the
set-couplings, i.e.,
$\rho(s,t) \in \mathcal{R}(\delta(s), \delta(t))$, and
$F\colon \mathcal{D}(S)\times \mathcal{D}(S)\to \mathcal{D}(S\times
S)$ the probabilistic couplings. This exactly corresponds to the
notion of \emph{coupling structure}
% as defined
in~\cite{bblmtb:prob-bisim-distance-automata-journal}.

Also the step in item~(\ref{fi:alg-min-fix-above:2c}) is analogous to
that in~\cite{bblmtb:prob-bisim-distance-automata-journal}. In fact,
the post-fixpoints of $\mathcal{M}_\#^{\mu \mathcal{M}_{C_i}}$ have been
shown in~\cite{bekp:fixpoint-theory-upside-down} to be the self-closed relations
of~\cite{bblmtb:prob-bisim-distance-automata-journal}
and thus verifying the emptiness of the greatest fixpoint corresponds
exactly to checking whether the largest self-closed relation is empty.

A difference concerns how strategy updates are performed.  While in
the algorithm derived above all set-couplings are updated at the same
time ($\rho$-component),
in~\cite{bblmtb:prob-bisim-distance-automata-journal} the set-coupling
is updated only for a single pair of states.
% (the update of the probabilistic
% couplings, the $f$-component, is instead global in both cases).
%
Since the ``local'' update produces a min-improvement, also the
algorithm in~\cite{bblmtb:prob-bisim-distance-automata-journal} can be
seen as an instance of the algorithm in
Fig.~\ref{fi:alg-min-fix-above}.
Updating all components can be more expensive, but it might
accelerate convergence. A more precise comparison should be carried
out via an experimental approach.
\end{toappendix}

\section{Conclusion}
\label{se:conclusion}

We developed abstract algorithms for strategy iterations which allow
to compute least fixpoints (or, dually, greatest fixpoints) of
non-expansive functions over MV-algebras. The idea consists in
expressing the function of interest as a minimum (or a maximum), and
view the process of computing the function as a game between players
{\Min} and {\Max} trying to minimise and maximise, respectively, the
outcome. Then the algorithms proceed via a sequence of steps which
converge to the least fixpoint from above, progressively improving the
strategy of player {\Min}, or from below, progressively improving the
strategy of the player {\Max}.
The two procedures have similar worst-case complexity. The number of
iterations is bounded by the number of strategies of the corresponding
player $p \in \{{\min},{\max}\}$, which is exponential in the input
size (the number of strategies is $\prod_{y\in Y}\len{H_p(y)}$). This
suggests that, depending on the setting, the fastest algorithm is the
one using the smaller decomposition $H_{\min}$ respectively
$H_{\max}$. However, a deeper analysis is still needed, as a smaller
decomposition usually leads to a higher cost for computing $\mu f_C$.

The algorithms generalise an approach
which has been recently proposed for simple stochastic games
in~\cite{bekp:fixpoint-theory-upside-down,DBLP:journals/corr/abs-2101-08184}.
We showed how our technique instantiates to energy games, thus giving
a method for determining the optimal strategies of both players,
and to the computation of the
behavioural distance for probabilistic automata, resulting in an
algorithm similar to the non-trivial procedure
in~\cite{bblmtb:prob-bisim-distance-automata-journal}, which was also
a source of inspiration.

Strategy iteration is used in many different application
domains with fairly similar underlying ideas and we believe that it is
fruitful to provide a general definition of the technique, clarifying and solving several issues on
this level, such as the need for stable improvements or ways to deal
with non-unique fixpoints.

There is an extremely wide literature on strategy iteration, often
also referred to as policy iteration or strategy improvement (for an
overview see \cite{GTW:book}). As mentioned in the introduction, after
its use on nonterminating stochastic
games~\cite{hk:nonterminating-stochastic-games}, it has been applied
to solve many kinds of games, including discounted mean-payoff
games~\cite{ZP:CMPGG}, parity games~\cite{VJ:DSIM,Sche:OSMPG} and
simple stochastic games~\cite{condon92}. Several quasi-polynomial
algorithms have been recently devised for parity
games~\cite{CJKLS:DPGQPT,jl:success-progress-measures,l:modal-mu-parity-quasi-polynomial}, while the existence of a polynomial algorithm is still
an important open problem. This has been generalized to finite
lattices by \cite{DBLP:conf/tacas/0001S21}.

Various papers on strategy iteration focus on lower
bounds~\cite{Fried:ELBLD,Fearnley10}. Our paper, rather than
concentrating on complexity issues, provides a general framework
capturing strategy iteration in a general lattice theoretical
setting. A work similar in spirit is~\cite{ABS:GSIM} which proposes a
meta-algorithm GSIA such that a number of strategy improvement
algorithms for SSGs arise as instances, along with a
general complexity bound. Differently from ours, this paper focuses on SSGs and iteration from below. However, it allows for the
parametrisation of the algorithm on a subset of edges of interest in
the game graph, which is not possible in our approach, and so it
can provide interesting suggestions for further generalisations.

Another interesting setting of application is the
lower-weak-upper-bound problem in mean-payoff games~\cite{BFLMS:2008},
reminiscent of energy games. For this problem, differently from the usual
definition, the aim for one player is to maximise, never going
negative, some resource which cannot exceed a given bound, while the
other player has to minimise it.  Also in this case, the solution can
be computed as a least fixpoint.  Due to the upper bound imposed to
the resource, the function is not non-expansive, thus it is not
captured by our theory.  Still, the algorithm KASI proposed
in~\cite{BrimChal:2010}, which computes the solution via strategy
iteration, shares many similarities with our approach from below: at
each iteration the algorithm computes a stable max-improvement of the
current strategy.  Indeed, when applying KASI to the special case
where there is no upper bound to the accumulated resource, called
lower-bound problem in~\cite{BFLMS:2008} (also studied under different
names in~\cite{CdAHS:2003,LifsPavl:2006}), the algorithm comes out as
an exact instantiation of our general strategy iteration from below.

Given their generality, we believe that the algorithms proposed in
the present paper have the potential to be applicable to a variety of other
settings. In particular, some preliminary investigations show their
applicability to computing behavioural metrics in an abstract
coalgebraic setting~\cite{bbkk:coalgebraic-behavioral-metrics}. Here
the behavioural distance is naturally characterised as a least
fixpoint of an operator based on the Wasserstein lifting of the
behavioural functor. Then the idea is to view couplings used in the
computation of the Wasserstein lifting as strategies and use strategy
iteration for converging to the coalgebraic metric.

Our abstract strategy iteration algorithms rely on the assumption that, once a strategy for one of the players is fixed, the optimal ``answering'' strategy for the opponent can be computed efficiently. Identifying abstract settings where
a min- or max-decompositions of a function ensures that
the answering strategy can be indeed computed efficiently (e.g., via linear programming as it happens for simple stochastic games), is an interesting direction of future research.

\bibliographystyle{plain}
% \bibliography{editreferences}
\bibliography{references}

\end{document}